\documentclass[runningheads]{llncs}

\usepackage[latin9]{inputenc}
\usepackage{amssymb,amsmath}
\usepackage{graphicx}
\usepackage{url}
\usepackage{tabularx}
\usepackage{tikz}
\usetikzlibrary{calc}
\usetikzlibrary{decorations.markings}
\usepackage{hyperref}
\usepackage{booktabs}
\usepackage{enumitem}
\usepackage{color}
\usepackage{xspace}
\usepackage{rotating}

\usepackage{fullpage}

\usepackage{verbatim}

\global\long\def\P{\mathcal{P}}

\newcommand{\aw}[1]{#1}

\begin{document}


\title{Scheduling Meets Fixed-Parameter Tractability}
\date{$\,$}
\titlerunning{Scheduling Meets Fixed-Parameter Tractability}

\author{     Matthias Mnich\inst{1} \and Andreas Wiese\inst{2}}
\institute{Cluster of Excellence MMCI,
               Saarbr\"{u}cken, Germany.
               \email{m.mnich@mmci.uni-saarland.de}
        \and Max Planck Institute for Computer Science,
               Saarbr\"{u}cken, Germany.
               \email{awiese@mpi-inf.mpg.de}
               }

\authorrunning{Mnich, Wiese}
\maketitle

\pagestyle{plain}
\begin{abstract}

Fixed-parameter tractability analysis and scheduling are two core
domains of combinatorial optimization which led to deep understanding
of many important algorithmic questions. However, even though fixed-parameter
algorithms are appealing for many reasons, no such algorithms are
known for many fundamental scheduling problems.

\quad In this paper we present the first fixed-parameter algorithms for
classical scheduling problems such as makespan minimization, scheduling with job-dependent 
cost functions---one important example being weighted flow time---and scheduling
with rejection. To this end, we identify
crucial \mbox{parameters} that determine the problems' complexity. 
In particular, we manage to cope with the problem complexity stemming
from numeric input values, such as job processing times, which is usually
a core bottleneck in the design of fixed-parameter algorithms. We
complement our algorithms with $\mathsf{W[1]}$-hardness results showing
that for smaller sets of parameters the respective problems do not
allow $\mathsf{FPT}$-algorithms. In particular, our positive and
negative results for scheduling with rejection explore a research direction proposed
by D{\'a}niel Marx~\cite{Marx2011}. 

\quad We hope that our contribution yields a new and fresh perspective on
scheduling and fixed-parameter algorithms and will lead to further
fruitful interdisciplinary research connecting these two areas. 
\end{abstract}

\section{Introduction}

\label{sec:introduction} Scheduling and fixed-parameter tractability
are two very well-studied research areas. In scheduling, the usual
setting is that one is given a set of machines and a set of jobs
with individual characteristics. The jobs need to be scheduled on
the machines according to some problem-specific constraints, such as
release dates, precedence constraints, or rules regarding preemption
and migration. Typical objectives are minimizing the global makespan,
the weighted sum of completion times of the jobs, or the total flow time.
During the last decades of research on scheduling, many important
algorithmic questions have been settled. For instance, for minimizing
the makespan and the weighted sum of completion time on identical
machines, $(1+\epsilon)$-approximation algorithms (PTASs) are known
for almost all $\mathsf{NP}$-hard settings~\cite{Afrati1999,HochbaumShmoys1987}.

However, the running time of these approximation schemes usually has
a bad dependence on~$\epsilon$, and in practise exact algorithms
are often desired. These considerations motivate to study which scheduling
problems are \emph{fixed-parameter tractable} ($\mathsf{FPT}$), which
amounts to identifying instance-dependent parameters~$k$ that allow
for algorithms that find optimal solutions in time~$f(k)\cdot n^{O(1)}$
for instances of size $n$ and some function $f$ depending only on~$k$.
Separating the dependence of $k$ and $n$ is often much more desirable
than a running time of, e.g.,~$O(n^{k})$, which becomes infeasible even
for small $k$ and large $n$. The parameter~$k$ measures the complexity
of a given instance and thus, problem classification according to
parameters yields an instance-depending measure of problem hardness.

Despite the fundamental nature of scheduling problems, and the clear
advantages of fixed-parameter algorithms, to the best of our knowledge
no such algorithms are known for the classical scheduling problems we study here. One obstacle
towards obtaining positive results appears to be that---in contrast
to most problems known to be fixed-parameter tractable---scheduling
problems involve many numerical input data (e.g., job processing times,
release dates, job weights), which alone render many problems $\mathsf{NP}$-hard,
thus ruling out fixed-parameter algorithms.
One contribution of this paper is that---for the fundamental problems studied 
here---choosing the number of distinct numeric values or an upper bound 
on them as the parameter suffices to overcome this impediment. Note that this requirement
is much weaker than assuming the parameter to be bounded by a constant (that can appear
in the exponent of the running time), as e.g., a running time of $2^{O(k)} \cdot poly(n)$
implies an \emph{exact} polynomial time algorithm for $k\in O(\log n)$ 
which is not true if the running time is $O(n^k)$.


We hope that our work gives rise to a novel perspective on scheduling
as well as on $\mathsf{FPT}$, yielding further interesting research
on fixed-parameter tractability of the many well-established and important
scheduling problems.


\subsection{Our Contributions}

\label{sec:ourcontributions} In this paper we present the first fixed-parameter
algorithms for several fundamental scheduling problems. In Section~\ref{sec:makespanminimization}
we study one of the most classical scheduling problems, which is minimizing
the makespan on an arbitrary number of machines without preemption,
i.e. the problem $P||C_{\max}$. Assuming integral input data, our
parameter $k$ defines an upper bound on the job processing times
appearing in an instance with $n$ jobs. We first prove that for any number of machines, we can restrict ourselves
to (optimal) solutions where jobs of the same length are almost equally
distributed among the machines, up to an additive error term of $\pm f(k)$
jobs. This insight can be used as an independent preprocessing routine
which optimally assigns the majority of the jobs of an instance (given
that $n\gg k$). After this preparation, we show that the remaining
problem can be formulated as an integer program in fixed dimension,
yielding an overall running time of $f(k)\cdot n^{O(1)}$. We note
that without a fixed parameter, the problem is strongly $\mathsf{NP}$-hard. 
For the much more general machine model of unrelated machines, we show that
$R||C_{\max}$ is fixed-parameter tractable when choosing the number of machines and the number of
distinct processing times as parameters. 
We reduce this problem again to integer programming in fixed dimension where our variables model
how many jobs of each size are scheduled on each machine. To ensure that an assignment of
all given jobs to these ``slots'' exists we argue via Hall's Theorem and ensure that for each subset of jobs
there are enough usable slots. While there the number of such subsets is exponentially large,
we identify a compact set of constraints that alone ensure that for the computed slots on the machines
a feasible job assignment exists. 
Let us remark that our problem is general enough that we do not see a way of using the ``number of numbers'' result by
Fellows et al.~\cite{FellowsEtAl2012}.
Note that if the number of machines or the number of processing times are constant, the problem
is still $\mathsf{NP}$-hard~\cite{lenstra1990approximation}, and thus no $\mathsf{FPT}$-algorithms can exist for those cases.


%



Then, in Section~\ref{sec:minimizingweightedcompletiontime}, we study
scheduling with rejection. Each job $j$ is specified by a processing
time~$p_{j}$, a weight $w_{j}$, and a rejection cost $e_{j}$ (all
jobs are released at time zero). We want to reject a set~$J'$ of
at most~$k$ jobs, and schedule all other jobs on one machine to minimize
$\sum_{j\notin J'}w_{j}C_{j}+\sum_{j\in J'}e_{j}$. For a given instance,
we identify three key parameters: the number of distinct processing
times, the number of distinct weights, and the maximum number $k$
of jobs to be rejected. We show that if any two of the three values
are taken as parameters, the problem becomes fixed-parameter tractable.
If $k$ and either of the other two are parameters, then we show that
an optimal solution is characterized by one of sufficiently few possible
patterns of jobs to be rejected. Once we guessed the correct pattern,
an actual solution can be found by a dynamic program efficiently.
If the number of distinct processing times and lengths are parameters
(but not~$k$), we provide a careful modeling of the problem as an
integer program with convex objective function in fixed dimension.
Here, we need to take particular care to incorporate the rejection
costs as bounded-degree polynomials, to be able to use the deep theory
of solving convex programs in fixed dimension efficiently.
To the best of our knowledge, this is the first time that convex programming is used in fixed-parameter algorithms.
We complement this result by showing that if only the number of rejected jobs $k$ is
the fixed parameter, then the problem becomes $\mathsf{W}[1]$-hard,
which prohibits the existence of a fixed-parameter algorithm, unless
$\mathsf{FPT}=\mathsf{W}[1]$ (which would imply subexponential time algorithms for many canonical $\mathsf{NP}$-complete problems). In particular, this justifies choosing
not only~$k$ but also the number of distinct processing times or
weights as a parameter. Our results respond to a quest by Marx~\cite{Marx2011}
for investigating the fixed-parameter tractability
of scheduling with rejection.

Finally, in Section~\ref{sec:schedulekjobswithelllengths} we turn
our attention to the parametrized dual of the latter problem: scheduling
with rejection of at least $n-s$ jobs ($s$ being the parameter). We reduce this to a much more
general problem which can be cast as the profit maximization version
of the \emph{General Scheduling Problem (GSP)}~\cite{BansalPruhs2010}.
We need to select a subset $J'$ of at most $s$ jobs to schedule
from a given set $J$, and each scheduled job~$j$ yields a profit
$f_{j}(C_{j})$, depending on its completion time~$C_{j}$. Note that
this function can be different for each job. Additionally, each job
$j$ has a release date $r_{j}$ and a processing time $p_{j}$. The
goal is to schedule these jobs on one machine to maximize $\sum_{j\in J'}f_{j}(C_{j})$.
We study the preemptive as well as the non-preemptive version of this
problem. In its full generality, GSP is not well understood. Despite
that, we are able to give a fixed-parameter algorithm if the number
of distinct processing times is bounded by a parameter, as well as
the maximum cardinality of~$J'$. Also, this implies that we are
able to solve scheduling with rejection of at least $n-s$ jobs, 
for any job dependent profit function. In particular, this includes functions that stem from difficult scheduling objectives
such as weighed flow time, weighted tardiness, and weighted flow time
squared. We complement our findings by showing that if only the number
of scheduled jobs $s$ is a fixed parameter, the problem is $\mathsf{W}[1]$-hard.
This justifies to consider additionally the number of distinct processing
times as a parameter. On the other hand, if only the latter is the parameter
the problem is para-$\mathsf{NP}$-hard, as it is $\mathsf{NP}$-hard if $p_j \in \{1,3\}$
for each job~$j$. In particular, we prove that for this setting the ordinary GSP
is $\mathsf{NP}$-hard, which might be of independent interest.

\begin{table}

\setlength{\extrarowheight}{3pt}
\centering %
\begin{tabular}{ccc}
\toprule 
Problem & Parameters & Result\tabularnewline
\midrule
$P||C_{\max}$  & maximum $p_{j}$ & $\mathsf{FPT}$\tabularnewline

$R||C_{\max}$  & \#distinct $p_{j}$ and \#machines & $\mathsf{FPT}$ \tabularnewline

$1|r_{j},(pmtn)|\max\sum_{\le s}f_{j}(C_{j}$) & \#selected jobs $s$ and \#distinct $p_{j}$ & $\mathsf{FPT}$\tabularnewline

$1|r_{j},(pmtn)|\max\sum_{\le s}f_{j}(C_{j}$) & \#selected jobs $s$ & $\mathsf{W[1]}$-hard\tabularnewline

$1|r_{j},(pmtn)|\max\sum_{\le s}f_{j}(C_{j}$) & \#distinct $p_{j}$ (in fact $\forall\, p_{j}\in\{1,3\}$) & para-$\mathsf{NP}$-hard\tabularnewline

$1||\sum_{\le k} e_{j}+\sum w_{j}C_{j}$  & \#rejected jobs $k$ and \#distinct $p_{j}$ & $\mathsf{FPT}$\tabularnewline

$1||\sum_{\le k} e_{j}+\sum w_{j}C_{j}$  & \#rejected jobs $k$  and \#distinct $w_{j}$ & $\mathsf{FPT}$\tabularnewline

$1||\sum_{\le k} e_{j}+\sum w_{j}C_{j}$  & \#distinct $p_{j}$ and \#distinct $w_{j}$ & $\mathsf{FPT}$\tabularnewline

$1||\sum_{\le k} e_{j}+\sum w_{j}C_{j}$  & \#rejected jobs $k$  & $\mathsf{W[1]}$-hard\tabularnewline
\bottomrule
\vspace{-1ex}
\end{tabular}

\caption{Summary of our results. 
 }

\label{tab:results} 
\end{table}

\vspace*{-1em}

Our contributions are summarized in Table~\ref{tab:results}.

%
%



\subsection{Related Work}

\label{sec:relatedwork} 
\textbf{Scheduling.} One very classical scheduling problem studied
in this paper is to schedule a set of jobs non-preemptively on a set
of $m$ identical machines, i.e.,~$P||C_{\max}$. Research for
it dates back to the 1960s when Graham showed that the greedy list
scheduling algorithm yields a $(2-\frac{1}{m})$-approximation and
a $4/3$-approximation when the jobs are ordered non-decreasingly
by length~\cite{Graham69}. After a series of improvements~\cite{CoffmanGareyJohnson78,Friesen84,Langston81,Sah76},
Hochbaum and Shmoys present a polynomial time approximation scheme~(PTAS), even if the number of machines is part of the input~\cite{HochbaumShmoys1987}.
On unrelated machines, the problem is $\mathsf{NP}$-hard to approximate
with a better factor than $3/2$~\cite{EbenlendrEtAl2008,lenstra1990approximation}
and there is a 2-approximation algorithm~\cite{lenstra1990approximation} that
extends to the generalized assignment problem~\cite{shmoys1993approximation}.
For the restricted assignment case, i.e., each job has a fixed processing time 
and a set of machines where one can assign it to, Svensson~\cite{Svensson2012} 
gives a polynomial time algorithm that estimates the optimal makespan up to a
factor of $33/17 + \epsilon \approx 1.9412 + \epsilon$.

For scheduling jobs with release dates preemptively on one machine,
a vast class of important objective functions is captured by the General
Scheduling Problem (GSP). In its full generality, it is studied by
Bansal and Pruhs~\cite{BansalPruhs2010} who
give a $O(\log\log nP)$-approximation, improving on several earlier results for special cases. One particularly important special
case is the weighted flow time objective where previously to Bansal and Pruhs~\cite{BansalPruhs2010}
the best known approximation factors where $O(\log^{2}P)$, $O(\log W)$, and $O(\log nP)$~\cite{BansalDhamdhere2007,ChekuriKhannaZhu2001}; here, $P$ and $W$ are the maximum ratios of job processing times and weights, respectively. Also, a quasi-PTAS with running
time $n^{O(\log P\log W)}$ is known~\cite{ChekuriKhanna2002}. 

A generalization of classical scheduling problems is \emph{scheduling
with rejection}. There, each job~$j$ is additionally equipped with
a \emph{rejection cost $e_{j}$.} The scheduler has the freedom to
reject job~$j$ and to pay a penalty of $e_{j}$, in addition to
some (ordinary) objective function for the scheduled jobs. For one
machine and the objective being to minimize the sum of weighted completion
times, Engels et al.~\cite{EngelsEtAl2003} give an optimal pseudopolynomial
dynamic program for the case that all jobs are released at time zero
and show that the problem is weakly $\mathsf{NP}$-hard. Sviridenko
and Wiese~\cite{SviridenkoWiese2013} give a PTAS for arbitrary release
dates. For the scheduling objective being the makespan and given multiple
machines, Hoogeveen~et~al.~\cite{HoogeveenEtAl2003} give \mbox{FPTASs}
for almost all machine settings, and a $1.58$-approximation for the
$\mathsf{APX}$-hard case of an arbitrary number of unrelated machines
(all results when allowing preemption).

\medskip{}
\noindent
\textbf{Fixed-Parameter Tractability.} Until now, to the best of our
knowledge, no fixed-parameter algorithms for the classical scheduling
problems studied in this paper have been devised. Though restricted scheduling
problems have been considered by the parameterized complexity community,
this generally meant that jobs are represented as vertices of a graph
with conflicting jobs connected by an edge, and then one finds a maximum
independent set in the graph; examples are given by van Bevern et al.~\cite{vanBevernEtAl2012}.
In contrast, classical scheduling problems investigated in the framework
of parameterized complexity appear to be intractable; for example,
$k$-processor scheduling with precedence constraints is $\mathsf{W}[2]$-hard~\cite{BodlaenderFellows1995}
and scheduling unit-length tasks with deadlines and precedence constraints
and~$k$ tardy tasks is $\mathsf{W}[1]$-hard~\cite{FellowsMcCarthy2003}, for parameter~$k$.

Mere exemptions seem to be an algorithm by Marx and Schlotter~\cite{MarxSchlotter2011}
for makespan minimization where~$k$ jobs have processing time $p\in\mathbb{N}$
and all other jobs have processing time 1, for combined parameter
$(k,p)$, and work by Chu et al.~\cite{ChuEtAl2013} who consider
checking \emph{feasibility} of a schedule (rather than optimization).


A potential reason for this lack of positive results (fixed-parameter
algorithms) might be that the knowledge of fixed-parameter algorithms
for \emph{weighted} problems is still in a nascent stage, and scheduling
problems are inherently weighted, having job processing times, job
weights, etc. We remark though that some scheduling-type problems can be 
addressed by choosing as parameter the ``number of numbers'', as done by Fellows et al.~\cite{FellowsEtAl2012}.

\section{Minimizing the Makespan}
\label{sec:makespanminimization} In this section we consider the
problem $P||C_{\max}$, i.e., where a given a set~$J$ of $n$ jobs
(all released at time zero) must be scheduled non-preemptively on
a set of~$m$ identical machines, as to minimize the makespan
of the schedule. 
We develop a fixed-parameter algorithm that solves this problem in
time $f(p_{\max})\cdot n^{O(1)}$, so our parameter is the maximum
processing time $p_{\max}$ over all jobs. 

In the sequel, we say that some job $j$ is of \emph{type $t$} if
$p_{j}=t$; let $J_{t}:=\{j\in J~|~p_{j}=t\}$. 
First, we prove that there is always an optimal solution in which
each machine has almost the same number of jobs of each type, up to
an additive error of $\pm f(p_{\max})$ for suitable function $f$.
This allows us to fix some jobs on the machines. For the remaining
jobs, we show that each machine receives at most $2f(p_{\max})$ jobs of
each type; hence there are only $(2f(p_{\max}))^{p_{\max}}$ possible
configurations for each machine. We solve the remaining problem with
an integer linear program in fixed dimension.


As a first step, for each type $t$, we assign $\left\lfloor \frac{|J_{t}|}{m}\right\rfloor -f(p_{\max})$
jobs of type $t$ to each machine; let $J_0\subseteq J$ be this set of jobs. This is justified by the next
lemma, which can be proven by starting with an arbitrary optimal schedule and exchanging jobs carefully between machines until the claimed property holds.

\begin{lemma}\label{lem:equal-distribution} There is a function
$f:\mathbb{N}\rightarrow\mathbb{N}$ with $f(p_{\max}) = 2^{O(p_{\max}\cdot\log p_{\max})}$ for all $p_{\max}\in \mathbb{N}$ 
 such that every instance of $P||C_{\max}$ admits an optimal solution
in which for each type $t$, each of the $m$ machines schedules at
least $\left\lfloor |J_{t}|/m\right\rfloor -f(p_{\max})$
and at most $\left\lfloor |J_{t}|/m\right\rfloor +f(p_{\max})$
jobs of type $t$.
\end{lemma}
\begin{proof}
For each machine $i$ denote by $J_{t}^{(i)}$ the set of jobs of
type $t$ scheduled on $i$ in some (optimal) schedule. We prove the
following (more technical) claim: there always exists an optimal solution
in which for every pair of machines $i$ and $i'$, and for each $\ell\in\{1,\hdots,p_{\max}\}$,
we have that $||J_{\ell}^{(i)}|-|J_{\ell}^{(i')}||\leq h(\ell)\cdot g(p_{\max})$,
where $g(p_{\max}):=(p_{\max})^{3}+p_{\max}$ and $h(\ell)$ is inductively defined by setting
$h(p_{\max}):=1$ and $h(\ell):=1+\sum_{j=\ell+1}^{p_{\max}}j\cdot h(j)$.

Suppose that there are two machines $i,i'$ such that there is a value
$\ell$ for which $||J_{\ell}^{(i)}|-|J_{\ell}^{(i')}||>h(\ell)\cdot g(p_{\max})$.
Assume, without loss of generality, that $|J_{\ell}^{(i)}|-|J_{\ell}^{(i')}|>h(\ell)\cdot g(p_{\max})$
and $||J_{\ell'}^{(\bar{i})}|-|J_{\ell'}^{(\bar{i}')}||\leq h(\ell')\cdot g(p_{\max})$
for all $\ell'>\ell$ and all machines $\bar{i},\bar{i}'$. We show
that $i'$ has at least a certain volume of jobs which are shorter
than $\ell$. We calculate that 
\begin{eqnarray*}
\sum_{j=1}^{\ell-1}j\cdot|J_{j}^{(i')}| & = & \left(\sum_{j=1}^{p_{\max}}j\cdot|J_{j}^{(i')}|\right)-\ell\cdot|J_{\ell}^{(i')}|-\left(\sum_{j=\ell+1}^{p_{\max}}j\cdot|J_{j}^{(i')}|\right)\\
 & \ge & \left(-p_{\max}+\sum_{j=1}^{p_{\max}}j\cdot|J_{j}^{(i)}|\right)-\ell\cdot|J_{\ell}^{(i')}|-\left(\sum_{j=\ell+1}^{p_{\max}}j\cdot|J_{j}^{(i')}|\right)\\
 & \ge & \left(-p_{\max}+\sum_{j=1}^{p_{\max}}j\cdot|J_{j}^{(i)}|\right)-\ell\cdot|J_{\ell}^{(i')}|-\left(\sum_{j=\ell+1}^{p_{\max}}j\cdot(|J_{j}^{(i)}|+h(j)\cdot g(p_{\max}))\right)\\
 & > & -p_{\max}+\sum_{j=1}^{\ell}j\cdot|J_{j}^{(i)}|+\ell\cdot(h(\ell)\cdot g(p_{\max})-|J_{\ell}^{(i)}|)-\sum_{j=\ell+1}^{p_{\max}}j\cdot h(j)\cdot g(p_{\max})\\
 & = & -p_{\max}+\sum_{j=1}^{\ell-1}j\cdot|J_{j}^{(i)}|+\ell\cdot h(\ell)\cdot g(p_{\max})-\sum_{j=\ell+1}^{p_{\max}}j\cdot h(j)\cdot g(p_{\max})\\
 & \ge & (p_{\max})^{3},
\end{eqnarray*}
where the first inequality stems from the fact that the considered
schedule is optimal and therefore the loads of any two machines can
differ by at most $p_{\max}$. Hence, there must be at least one type $\ell'<\ell$
such that $\ell'\cdot|J_{\ell}^{(i')}|\geq (p_{\max})^{2}$, and thus $|J_{\ell}^{(i')}|\geq p_{\max}$.
Since the least common multiple of any two values in $\{1,\hdots,p_{\max}\}$
is at most $(p_{\max})^{2}$, we can swap $r\in\{1,\hdots,p_{\max}\}$ jobs of type~$\ell$ from machine $M_{i}$ with $r'\in\{1,\hdots,p_{\max}\}$ jobs of
type $\ell'$ from machine $M_{i'}$, without changing the total load
of any of the two machines. By continuing inductively we obtain a
schedule satisfying $||J_{\ell'}^{(i)}|-|J_{\ell'}^{(i')}||\le h(\ell')\cdot g(p_{\max})$
for any two machines $i,i'$ and any type $\ell'\geq\ell$. Inductively,
we obtain an optimal schedule satisfying the claim for all types $\ell$.
Thus, the claim of the lemma follows for an appropriately chosen function
$f$.

We now investigate in more detail which choice of $f$ is suitable.
To this end, consider first the function $h$~as recursively defined above.
Notice that $h$ is parametrized by $\ell$, but also depends on $p_{\max}$.
We prove, by induction on $\ell$, that $h(\ell) = \frac{(p_{\max}+1)!}{(\ell + 1)!}$.
The base case is for $\ell = p_{\max}$, and it follows from the definition that $h(p_{\max}) = 1 = \frac{(p_{\max}+1)!}{(p_{\max}+1)!}$.
For the inductive step, let $\ell < p_{\max}$ and take as inductive hypothesis that $h(\ell') = \frac{(p_{\max}+1)!}{(\ell'+1)!}$ for all $\ell'\in\{\ell+1,\hdots,p_{\max}\}$.
We then have
\begin{align*}
  h(\ell) & = 1 + \sum_{j=\ell+1}^{p_{\max}}j\cdot h(j)
           = 1 + (\ell+1)h(\ell+1) + \sum_{j=\ell+2}^{p_{\max}}j\cdot h(j)\\
          & = 1 + (\ell+1)h(\ell+1) + (h(\ell+1) - 1)
           = (\ell+2)h(\ell+1)\\
          & = (\ell+2)\frac{(p_{\max}+1)!}{(\ell+2)!}
           = (\ell+2)\frac{(p_{\max}+1)!}{(\ell+1)!(\ell+2)}
           = \frac{(p_{\max}+1)!}{(\ell+1)!},
\end{align*}
as claimed.

Thus, the claim of the lemma holds for $f(p_{\max}) := h(1)\cdot g(p_{\max})$, and therefore we have $f(p_{\max}) \le ((p_{\max}+1)!)((p_{\max})^3 + p_{\max}) = 2^{O(p_{\max}\log p_{\max})}$.
\qed
\end{proof}

Denote by $J' = J\setminus J_0$ be the set of yet unscheduled jobs. \aw{We ignore all other jobs from now on.}
By Lemma~\ref{lem:equal-distribution},
there is an optimal solution in which each machine receives at most $2\cdot f(p_{\max})+1$ jobs from each type.
Hence, there are at most $(2\cdot f(p_{\max})+2)^{p_{\max}}$ ways how
the schedule for each machine can look like (up to permuting jobs
of the same length). Therefore, the remaining problem can be solved
with the following integer program. Define a set $\mathcal{C}=\{0,\hdots,2\cdot f(p_{\max})+1\}^{p_{\max}}$
of at most $(2\cdot f(p_{\max})+2)^{p_{\max}}$ ``configurations'',
where each \emph{configuration} is a vector $C\in\mathcal{C}$ encoding
the number of jobs from~$J'$ of each type assigned to a machine.

In any optimal solution for $J'$, the makespan is in the range $\{\left\lceil p(J')/m\right\rceil ,\linebreak\hdots,\left\lceil p(J')/m\right\rceil +p_{\max}\}$,
as $p_{j}\le p_{\max}$ for each $j$. For each value $T$ in this
range we try whether $\mathsf{opt}\le T$, where $\mathsf{opt}$ denotes the minimum makespan of the instance. So fix a value $T$. 
We allow only configurations $C=(c_{1},\hdots,c_{p_{\max}})$
which satisfy $\sum_{i=1}^{p_{\max}}c_{i}\cdot i\le T$; let
$\mathcal{C}(T)$ be the set of these configurations. We define an
integer program. For each $C\in\mathcal{C}(T)$, introduce a variable
$y_{C}$ for the number of machines with configuration~$C$ in the
solution. (As the machines are identical, only the number of machines
following each configuration is important.) 
\begin{eqnarray}
\sum_{C\in\mathcal{C}(T)}y_{C} & \leq & ~m\label{eq:bound-no-machines}\\
\sum_{C=(c_{1},\hdots,c_{p_{\max}})\in\mathcal{C}(T)}y_{C}\cdot c_{p} & \geq & ~\aw{|J'\cap J_{p}|},~p=0,\hdots,p_{\max}\label{eq:enough-slots}\\
y_{C} & \in & ~\{0,\hdots,m\},~C\in\mathcal{C}(T)\label{eq:config-feasible}
\end{eqnarray}
Inequality~(\ref{eq:bound-no-machines}) ensures that at most $m$
machines are used, inequalities~(\ref{eq:enough-slots}) ensure that
all jobs from each job type are scheduled. The above integer program~\eqref{eq:bound-no-machines}--\eqref{eq:config-feasible}
has at most $(2\cdot f(p_{\max})+2)^{p_{\max}}$ dimensions.

To determine feasibility of \eqref{eq:bound-no-machines}--\eqref{eq:config-feasible},
we employ deep results about integer programming in fixed dimension.
\aw{As we will need it later, we cite here an algorithm that can even minimize over
convex spaces described by (quasi-)convex polynomials, rather than only over polytopes.}
Known efficient algorithms for convex integer minimization in
fixed dimension use variants of Lenstra's algorithm~\cite{Lenstra1983}
for integer programming. 
The first algorithm of this kind was given by Khachiyan~\cite{Khachiyan1984},
later improved by Heinz~\cite{Heinz2005}. 
We use the algorithm by Heinz as stated in a survey by K{\"{o}}ppe~\cite[Theorem 6.4]{Koppe2012}.

\begin{theorem}[\cite{Heinz2005,Koppe2012}]
\label{thm:convexintegerminimizationisfpt} Let $f,g_{1},\hdots,g_{m}\in\mathbb{Z}[x_{1},\hdots,x_{t}]$
be quasi-convex polynomials of degree at most \mbox{$d\ge2$}, whose coefficients
have a binary encoding length of at most~$\ell$. There is an
algorithm that in time $m\cdot\ell^{O(1)}\cdot d^{O(t)}\cdot2^{O(t^{3})}$
computes a minimizer $\mathbf{x}^{\star}\in\mathbb{Z}^{t}$ of binary encoding size $\ell\cdot d^{O(t)}$
for following problem~\eqref{eqn:convexlenstra}, or reports that
no minimizer exists.
\begin{equation}
\min f(x_{1},\hdots,x_{t})\quad\textnormal{subject to}~g_{i}(x_{1},\hdots,x_{t})\leq0,~~i=1,\hdots,m\quad\mathbf{x}\in\mathbb{Z}^{t}\enspace.\label{eqn:convexlenstra}
\end{equation}
\end{theorem} 

%

Thus, we can determine
feasibility of \eqref{eq:bound-no-machines}--\eqref{eq:config-feasible}
in time $g(p_{\max})\cdot (\log n + \log m)^{O(1)}$, for suitable function~$g$.
The smallest value $T$ for which it is feasible gives the optimal makespan and together with the preprocessing routine due to Lemma~\ref{lem:equal-distribution}
this yields an optimal schedule. 

\begin{theorem} \label{thm:makespanmaxprocessingtime}
There is a function $f$ such that instances of $P||C_{\max}$ with $n$ jobs and $m$ machines can be solved in time~$f(p_{\max})\cdot (n + \log m)^{O(1)}$. \end{theorem}

Recall that without choosing a parameter, problem $P||C_{\max}$ is strongly $\mathsf{NP}$-hard
(as it contains \textsc{3-Partition}).
A natural extension to consider is the problem $P||C_{\max}$ parameterized
by the number $\overline{p}$ of distinct processing times. Unfortunately,
for this extension Lemma~\ref{lem:equal-distribution} is no longer
true. Let $q_{1},q_{2}$ be two different (large) prime numbers. Consider
an instance with two identical machines $M_{1},M_{2}$, and $q_{1}$
many jobs with processing time $q_{2}$, and similarly $q_{2}$ many
jobs with processing time $q_{1}$. The optimal makespan is $T:=q_{1}\cdot q_{2}$
which is achieved by assigning all jobs with processing time $q_{1}$
on $M_{1}$ and all other jobs on~$M_{2}$, giving both machines a
makespan of exactly $T$. However, apart from swapping the machines
this is the only optimal schedule since the equation $q_{1}\cdot q_{2}=x_{1}\cdot q_{1}+x_{2}\cdot q_{2}$
only allows integral solutions $(x_{1},x_{2})$ such that~$x_{1}$
is a multiple of $q_{2}$ and $x_{2}$ is a multiple of $q_{1}$.%

Indeed,
for constantly many processing times the problem was only recently shown to
be polynomial time solvable for any constant $\overline{p}$~\cite{GoemansRothvoss2014}.

\subsection{Bounded Number of Unrelated Machines}

We study the problem $Rm||C_{\max}$ where now the machines are unrelated,
meaning that a job can have different processing times on different
machines. In particular, it might be that a job cannot be processed
on some machine at all, i.e., has infinite processing time on that
machine. We choose as parameters the number~$\overline{p}$ of distinct
(finite) processing times and the number of machines $m$ of the instance.

We model this problem as an integer program in fixed dimension. Denote
by $q_{1},\hdots,q_{\overline{p}}$ the distinct finite processing
times in a given instance. For each combination of a machine~$i$ and a \aw{finite} processing
time $q_{\ell}$ we introduce a variable $y_{i,\ell}\in\{0,\hdots,n\}$
that models how many jobs of processing time $q_{\ell}$ are assigned
to~$i$. Note that the number of these variables is bounded by $m\cdot\overline{p}$.
An assignment to these variables can be understood as allocating $y_{i,\ell}$
slots for jobs with processing time $q_{\ell}$ to machine $i$, without
specifying what actual jobs are assigned to these slots. Assigning
the jobs to the slots can be understood as a bipartite matching: introduce
one vertex $v_{j}$ for each job $j$, one vertex $w_{s,i}$ for each
slot~$s$ on each machine $i$, and an edge $\{v_{j},w_{s,i}\}$ whenever
job $j$ has the same size on machine $i$ as slot $s$. According
to Hall's Theorem, there is a matching in which each job is matched
if and only if for each set of jobs $J'\subseteq J$ there are at
least~$|J'|$ slots to which at least one job in $J'$ can be assigned.
For one single set~$J'$ the latter can be expressed by a linear constraint;
however, the number of subsets~$J'$ is exponential. We overcome this
as follows: We say that two jobs $j,j'$ are of the same \emph{type
}if $p_{i,j}=p_{i,j'}$ for each machine~$i$. Note that there are
only $(\bar{p}+1)^{m}$ different types of jobs. As we will show,
it suffices to add a constraint for sets of jobs $J'$ such that for
each job type either all or none of the jobs of that type are contained
in $J'$ (those sets ``dominate'' all other sets). This gives rise
to the following integer program. Denote by $Z$ the set of all job
types and for each $z\in Z$ denote by $J_{z}\subseteq J$ the set
of jobs of type $z$. For each set $Z'\subseteq Z$ denote by~$Q_{i,Z'}$
the set of distinct finite processing times of jobs of types in~$Z'$
on machine $i$.
\begin{align}
\min~T\quad\textnormal{s.t.}~\sum_{\ell\in\{1,\hdots,\overline{p}\}}y_{i,\ell}\cdot q_{\ell} & \leq T, &  & i=1,\hdots,m\label{eq:R_Cmax-top}\\
\sum_{z\in Z'}|J_{z}| & \le\sum_{i}\sum_{\ell:q_{\ell}\in Q_{i,Z'}}y_{i,\ell} &  & \forall~Z'\subseteq Z\\
y_{i,\ell} & \in\{0,\hdots,n\} &  & i=1,\hdots,m,~\ell=1,\hdots,\overline{p}\\
T & \ge0\label{eq:R_Cmax-bottom}
\end{align}
As the above IP has only $m\cdot\bar{p}$ dimensions and $2^{(\bar{p}+1)^{m}}+m$
constraints, we solve it using Theorem~\ref{thm:convexintegerminimizationisfpt}. Finding then
a bipartite matching for the actual assignment of the jobs can be
done in polynomial time. This yields the following theorem.

\begin{theorem}\label{thm:makespan-R-no-pjs-no-machines} Instances
of $R||C_{\max}$ with $m$ machines and $n$ jobs with $\overline{p}$
distinct finite processing times can be solved in time~$f(\overline{p},m)\cdot(n + \log\max_{\ell}q_{\ell})^{O(1)}$
for a suitable function $f$. \end{theorem}
\begin{proof}We show that the integer program \eqref{eq:R_Cmax-top}-\eqref{eq:R_Cmax-bottom} has a solution of value $T$
if and only if there is a solution with makespan $T$. First suppose
that there is a solution with makespan $T$. Then for each machine
$i$ and each processing time $q_{\ell}$ we set $y_{i,\ell}$ to
be the number of jobs with processing time~$q_{\ell}$ on machine~$i$. As the makespan is $T$, we have that $\sum_{\ell\in\{1,\hdots,\overline{p}\}}y_{i,\ell}\cdot q_{\ell}\le T$
for each machine $i$. Now consider a constraint for a set $Z'\subseteq Z$.
All jobs in $\cup_{z\in Z'}J_{z}$ are assigned in the given solution.
Therefore, we must have that $\sum_{z\in Z'}|J_{z}|\le\sum_{i}\sum_{\ell:q_{\ell}\in Q_{i,Z'}}y_{i,\ell}$
since the right-hand side of the inequality denotes the number of
slots (given by the $y_{i,\ell}$ variables) that at least one job
in $\cup_{z\in Z'}J_{z}$ can be assigned to. 

Conversely, suppose that the IP \eqref{eq:R_Cmax-top}-\eqref{eq:R_Cmax-bottom} has a solution of value $T$.
We prove that in polynomial time we can construct a feasible schedule
with makespan $T$. We construct a bipartite graph as described above:
we introduce a vertex $v_{j}$ for each job $j$. For each machine
$i$ and each processing time $q_{\ell}$, we introduce $y_{i,\ell}$
slots of size $q_{\ell}$ and add a vertex $w_{s,i}$ for
each slot $s$ on machine $i$. For each pair of a job $j$ and a
slot $s$ on a machine~$i$, we introduce an edge $\{v_{j},w_{s,i}\}$
if and only if $p_{i,j}$ equals the size of slot $s$. According
to Hall's Theorem, there is a matching that matches each vertex $v_{j}$
if and only if for each set of jobs $J'\subseteq J$ there are at
least~$|J'|$ slots to which at least one job in $J'$ can be assigned.
Consider a set $J'\subseteq J$ and let $\mathrm{slot}(J')$ be the
set of slots that a job in $J'$ can be assigned to. Denote by $\mathrm{full}(J')$
the set of all jobs $j\in J$ for which there is a job $j'\in J'$
such that $j$ and $j'$ are of the same type. Observe that $J'\subseteq\mathrm{full}(J')$
and $\mathrm{slot}(J')=\mathrm{slot}(\mathrm{full}(J'))$. Then there
is a set $Z'\subseteq Z$ such that $\cup_{z\in Z'}J_{z}=\mathrm{full}(J')$.
Due to the constraint $\sum_{z\in Z'}|J_{z}|\le\sum_{i}\sum_{\ell:q_{\ell}\in Q_{i,Z'}}y_{i,\ell}$
there are enough slots for the jobs in $\cup_{z\in Z'}J_{z}=\mathrm{full}(J')$.
As $J'\subseteq\mathrm{full}(J')$ and $\mathrm{slot}(J')=\mathrm{slot}(\mathrm{full}(J'))$
there are also enough slots for the jobs in $J'$. Applying this reasoning
to all sets $J'\subseteq J$, due to Hall's Theorem there exists a
schedule with makespan at most $T$. Using standard bipartite matching
theory (see e.g.~\cite[Chapter 16]{Schrijver03}) we can find such
a matching in polynomial time.\qed\end{proof}

%
%
%

A natural question is the case when only the number of machines is a fixed parameter. If the
job processing times can be arbitrary, then already $P2||C_{\max}$ is $\mathsf{NP}$-hard due to 
the contained \textsc{Partition} problem.  
If one additionally requires the job processing times to be polynomially bounded, then $P||C_{\max}$ is still 
$\mathsf{W}[1]$-hard if only the number of machines is a parameter, 
due to a result by Jansen et al.~\cite{JansenEtAl2013}. 

On the other hand, $R||C_{\max}$ is 
$\mathsf{NP}$-hard if only processing times $\{1,2,\infty \}$ are allowed~\cite{EbenlendrEtAl2008,lenstra1990approximation}.
This justifies to take both $m$ and $\bar{p}$ as a parameters in the unrelated machine case.

\section{Scheduling with Rejection}

\label{sec:minimizingweightedcompletiontime} In this section we study
scheduling with rejection to optimize the weighted sum of completion
time plus the total rejection cost, i.e, $1||\sum_{\leq k}e_{j}+\sum w_{j}C_{j}$.
Formally, we are given an integer $k$ and a set $J$ of~$n$ jobs,
all released at time zero. Each job $j\in J$ is characterized by
a processing time $p_{j}\in\mathbb{N}$, a weight $w_{j}\in\mathbb{N}$
and rejection cost~$e_{j}\in\mathbb{N}$. The goal is to reject a
set $J'\subseteq J$ of at most $k$ jobs and to schedule all other
jobs non-preemptively on a single machine, as to minimize $\sum_{j\in J\setminus J'}w_{j}C_{j}+\sum_{j\in J'}e_{j}$.
Note that if rejection is not allowed ($k=0$), the problem is solved
optimally by scheduling jobs according to non-decreasing Smith ratios
$w_{j}/p_{j}$, breaking ties arbitrarily~\cite{Smith1956}.

\subsection{Parameters Number of Rejected Jobs and Distinct Processing Times
or Weights}

Denote by $\overline{p}\in\mathbb{N}$ the number of distinct processing
times in a given instance. First, we assume that~$\overline{p}$
and the maximum number $k$ of rejected jobs are parameters. Thereafter,
using a standard reduction, we will derive an algorithm for the case
that $k$ and the number $\overline{w}$ of distinct weights are parameters.

Denote by $q_{1},\hdots,q_{\overline{p}}$ the distinct processing
times in a given instance. For each $i\in\{1,\hdots,\overline{p}\}$,
we guess the number of jobs with processing time $q_{i}$ which are
rejected in an optimal solution. Each possible guess is characterized
by a vector $\mathbf{v}=\{v_{1},\hdots,v_{\bar{p}}\}$ whose entries
$v_{i}$ contain integers between $0$ and $k$, and whose total sum
is at most $k$. There are at most $(k+1)^{\overline{p}}$ such vectors
$\mathbf{v}$, each one prescribing that at most $v_{i}$ jobs of
processing time $p_{i}$ can be rejected. We enumerate them all. One
of these vectors must correspond to the optimal solution, so the reader may
\aw{assume} that we know this vector~$\mathbf{v}$. 

In the following, we will search for the optimal schedule that \emph{respects
$\mathbf{v}$}, meaning that for each $i\in\{1,\hdots,\overline{p}\}$
at most $v_{i}$ jobs of processing time $q_{i}$ are rejected. To
find an optimal schedule respecting~$\mathbf{v}$, we use a dynamic
program. 
Suppose the jobs in $J$ are labeled by $1,\hdots,n$ by non-increasing
\emph{Smith ratios}~$w_{j}/p_{j}$. Each dynamic programming cell
is characterized by a value $n'\in\{0,\hdots,n\}$, and a vector $\mathbf{v'}$
with $\overline{p}$ entries which is \emph{dominated} by $\mathbf{v}$,
meaning that $v'_{i}\le v_{i}$ for each $i\in\{1,\hdots,\overline{p}\}$.
For each pair $(n',\mathbf{v'})$ we have a cell $C(n',\mathbf{v'})$
modeling the following subproblem. Assume that for jobs in $J':=\{1,\hdots,n'\}$
we have already decided whether we want to schedule them or not. For
each processing time~$q_{i}$ denote by $n'_{i}$ the number of jobs
in $J'$ with processing time~$q_{i}$. Assume that for each type $i$,
we have decided to reject $v_{i}-v'_{i}$ jobs from~$J'$. Note that
then the total processing time of the scheduled jobs sums up to $t:=\sum_{i}q_{i}\cdot(n_{i}'-(v_{i}-v'_{i}))$.
It remains to define a solution for the jobs in $J'':=\{n'+1,\hdots,n\}$
during time interval $[t,\infty)$, such that for each type $i$
we can reject up to $v'_{i}$ jobs. %

Clearly, the cell $C(0,\mathbf{v})$ contains the optimal value to
the entire instance. Also, by definition, $C(n,\mathbf{v'})=0$ for
any $\mathbf{v'}$, since then $J''=\emptyset$. For the dynamic programming
transition, we prove:

\begin{lemma}
\label{thm:dptransition2} Let $C(n',\mathbf{v'})$ be a cell and
let $\mathsf{opt}(n',\mathbf{v'})$ be the optimal solution value
to its subproblem. Let $i\in\{1,\hdots,\overline{p}\}$ be such that
$p_{n'+1}=q_{i}$, and let $t:=\sum_{i}q_{i}\cdot(n_{i}'-(v_{i}-v'_{i}))$.
If $v'_{i}=0$~then\\
$$
\mathsf{opt}(n',(v'_{1},\hdots,v'_{i},\hdots,v'_{\overline{p}}))=\mathsf{opt}(n'+1,(v'_{1},\hdots,v'_{i},\hdots,v'_{\overline{p}}))+(t+p_{n'+1})\cdot w_{n'},
$$
otherwise
\begin{multline*}
\hspace*{-1em}~\mathsf{opt}(n',(v'_{1},\hdots,v'_{i},\hdots,v'_{\overline{p}}))=\\\min\Big\{\mathsf{opt}(n'+1,(v'_{1},\hdots,v'_{i},\hdots,v'_{\overline{p}}))+(t+p_{n'+1})\cdot w_{n'},
\mathsf{opt}(n'+1,(v'_{1},\hdots,v'_{i}-1,\hdots,v'_{\overline{p}}))+e_{n'+1}\Big\}.
\end{multline*}
\end{lemma}
\begin{proof}First,
suppose that $\mathsf{opt}(n',\mathbf{v'})$ schedules job $n'$.
Due to our sorting, $n'$ has the largest Smith ratio of all
jobs in $\{n',\hdots,n\}$, and an optimal schedule with value $\mathsf{opt}(n',\mathbf{v'})$
schedules~$n'$ at time $t$ at cost $(t+p_{n'})\cdot w_{n'}$, and
the optimal solution value to the remaining subproblem is given by
$\mathsf{opt}(n'+1,(v'_{1},\hdots,v'_{i},\hdots,v'_{\ell}))$. If a solution of value
$\mathsf{opt}(n',\mathbf{v'})$ rejects $n'$ (which can only happen
if $v'_{i}>0$), then\linebreak $\mathsf{opt}(n',\mathbf{v'})=\mathsf{opt}(n'+1,(v'_{1},\hdots,v'_{i}-1,\hdots,v'_{\ell}))+e_{n'}$.

On the other hand, there is a solution for cell $C(n',\mathbf{v'})$
given by scheduling~$n'$ at time $t$ and all other non-rejected
jobs in the interval $[t+p_{n'},\infty)$ at cost $\mathsf{opt}(n'+1,(v'_{1},\hdots,v'_{i},\hdots,v'_{\ell}))+(t+p_{n'})\cdot w_{n'}$.
If $v'_{i}>0$ there is also a feasible solution with cost $\mathsf{opt}(n'+1,(v'_{1},\hdots,v'_{i}-1,\hdots,v'_{\ell}))+e_{n'}$
which rejects job $n'$ and schedules all other non-rejected jobs
during the interval $[t,\infty)$.\qed \end{proof}

The size of the dynamic programming table is bounded
by $n\cdot(k+1)^{\overline{p}}$. Hence, the whole dynamic program
runs in fixed-parameter time, proving that: 
\begin{theorem} \label{thm:rejectjobswithboundednumberofprocessingtimes}
For sets $J$ of $n$ jobs with $\overline{p}$ distinct processing
time, the problem $1||\sum_{\leq k}e_{j}+\sum w_{j}C_{j}$ is solvable
in time $O(n\cdot(k+1)^{\overline{p}}+n\cdot\log n)$. \end{theorem}

For the problem $1||\sum w_{j}C_{j}$, it is known~\cite[Theorem 3.1]{ChudakHochbaum1999}
that one can interchange weights and processing times, bijectively
mapping feasible solutions preserving their costs. Hence, we obtain
for parameters $k$ and~$\overline{w}$ distinct job weights: \begin{corollary}
\label{thm:rejectjobswithboundednumberofweights} For sets $J$ of
$n$ jobs with $\overline{w}$ distinct job weights, the problem $1||\sum_{\leq k}e_{j}+\sum w_{j}C_{j}$
is solvable in time $O(n\cdot(k+1)^{\overline{w}}+n\cdot\log n)$.
\end{corollary}

We show that when only the number $k$ of rejected jobs is taken as parameter, problem becomes
$\mathsf{W}[1]$-hard. This justifies to define additionally the number of weights
or processing times as parameter.
We remark that when jobs have non-trivial release dates,
then even for $k=0$ the problem is $\mathsf{NP}$-hard~\cite{lenstra1977complexity}.


\begin{theorem}\label{thm:rejparamnumberofrejectedjobshardness}
Problem $1||\sum_{\leq k}e_{j}+\sum w_{j}C_{j}$ is $\mathsf{W}[1]$-hard
if \aw{the} parameter \aw{is} the number~$k$ of rejected jobs. \end{theorem}
\begin{proof}
We reduce from \textsc{$k$-Subset Sum}, for which the input consists of integers
$s_{1},\hdots,s_{n}\in\mathbb{N}$ and two values $k,q\in\mathbb{N}$.
The goal is to select a subset of $k$ of the given integers $s_{i}$
that sum up to~$q$. Parameterized by $k$, this problem is known
to be $\mathsf{W}[1]$-hard~\cite{FellowsKoblitz1993}. Our reduction
mimics closely a reduction from (ordinary)\textsc{ Subset Sum} in~\cite{EngelsEtAl2003}
that shows that $1||\sum e_{j}+\sum w_{j}C_{j}$ is weakly $\mathsf{NP}$-hard.

Suppose we are given an instance of $k$\textsc{-Subset Sum}. Let $S=\sum_{i=1}^{n}s_{i}$.
We construct an instance of $1||\sum_{\leq k}e_{j}+\sum w_{j}C_{j}$
with $n$ jobs, where each job $j$ has processing time $p_{j}:=s_{j}$,
weight $w_{j}:=s_{j}$, rejection cost $e_{j}:=(S-q)\cdot s_{j}+\frac{1}{2}s_{j}^{2}$. 
Since $p_{j}=w_{j}$ for all jobs $j$, the ordering of the scheduled
jobs $J\setminus J'$ does not affect the value of $\sum_{j\in J\setminus J'}w_{j}C_{j}$.
Using this fact and substituting for the rejection penalty, we can
rewrite the objective function as follows: 
\begin{eqnarray*}
\sum_{j\in J\setminus J'}w_{j}C_{j}+\sum_{j\in J'}e_{j} & = & \sum_{j\in J\setminus J'}s_{j}\sum_{i\leq j,i\in J\setminus J'}s_{i}+\sum_{j\in J'}e_{j}\\
 & = & \sum_{j\in J\setminus J'}s_{j}^{2}+\sum_{j<i,i,j\in J\setminus J'}s_{j}s_{i}+\sum_{j\in J'}\left((S-q)\cdot s_{j}+\frac{1}{2}s_{j}^{2}\right)\\
 & = & \frac{1}{2}\left[\left(\sum_{j\in J\setminus J'}s_{j}\right)^{2}+\sum_{j\in J\setminus J'}s_{j}^{2}\right]+(S-q)\sum_{j\in J'}s_{j}+\frac{1}{2}\sum_{j\in J'}s_{j}^{2}\\
 & = & \frac{1}{2}\left(\sum_{j\in J\setminus J'}s_{j}\right)^{2}+(S-q)\sum_{j\in J'}s_{j}+\frac{1}{2}\sum_{j=1}^{n}s_{j}^{2}\enspace.
\end{eqnarray*}
Since $\sum_{j=1}^{n}s_{j}^{2}$ does not depend on the choice of
$J'$, this is equivalent to minimizing the following function~$h(x)$,
with $x=\sum_{j\in J\setminus J'}s_{j}$: 
\begin{eqnarray*}
h(x):=\frac{1}{2}x^{2}+(S-q)(S-x)=\frac{1}{2}\left(\sum_{j\in J\setminus J'}s_{j}\right)^{2}+(S-q)\left(S-\sum_{j\in J\setminus J'}s_{j}\right).
\end{eqnarray*}
The unique minimum of $h(x)$ is $\frac{1}{2}S^{2}-\frac{1}{2}q^{2}$
at $x=S-q$, i.e., when $\sum_{j\in J\setminus J'}s_{j}=S-q$. Hence,
if such a set~$J'$ with $|J'|\le k$ exists, the resulting value
is optimal for the scheduling problem. Therefore, if the solution
to the created instance of $1||\sum_{\leq k}e_{j}+\sum w_{j}C_{j}$
has value less than or equal to $\frac{1}{2}S^{2}+\frac{1}{2}q^{2}+\frac{1}{2}\sum_{j=1}^{n}s_{j}^{2}$,
then the instance of \textsc{$k$-Subset Sum} is a ``yes''-instance.
Conversely, if the instance of \textsc{Subset Sum} is ``yes'', then
there is a schedule of value $\frac{1}{2}S^{2}+\frac{1}{2}q^{2}+\frac{1}{2}\sum_{j=1}^{n}s_{j}^{2}$.
\qed
\end{proof}

\subsection{Parameter Number of Distinct Processing Times and Weights}

\label{sec:parameternumberoftypes} We consider the number of distinct
processing times and weights as parameters. To this end, we say that
two jobs $j,j'$ are of the same \emph{type} if $p_{j}=p_{j'}$ and
$w_{j}=w_{j'}$; let $\tau$ be the number of types in an instance.
Note, however, that jobs with the same type might have different rejection
costs, so we cannot bound the ``number of numbers'' in the input by a function of the parameters only---this prohibits that we use a strategy similar to Fellows et al.~\cite{FellowsEtAl2012} who solve simpler problems parameterized by the number of distinct input numbers.
Instead, we resort to convex integer programming, which is used here for the first time in fixed-parameter algorithms.
The running time of our algorithm will depend only \emph{polynomially} on $k$, the upper
bound on the number of jobs we are allowed to reject. For each type
$i$, let $w^{(i)}$ be the weight and $p^{(i)}$ be the processing
time of jobs of type $i$. Assume that job types are numbered $1,\hdots,\tau$
such $w^{(i)}/p^{(i)}\ge w^{(i+1)}/p^{(i+1)}$ for each $i\in\{1,\hdots,\tau-1\}$.
Clearly, an optimal solution schedules jobs ordered non-increasingly
by Smith's ratio without preemption.

The basis for our algorithm is a problem formulation as a convex integer
minimization problem with dimension at most $2\tau$. In an instance,
for each $i$, we let $n_{i}$ be the number of jobs of type $i$
and introduce an integer variable $x_{i}\in\mathbb{N}_{0}$ modeling
how many jobs of type $i$ we decide to schedule. We introduce the
linear constraint $\sum_{i=1}^{\tau}(n_{i}-x_{i})\leq k,$ to ensure
that at most $k$ jobs are rejected.

The objective function is more involved. For each type $i$, scheduling the jobs of type~$i$ costs 
\begin{eqnarray*}
\sum_{\ell=1}^{x_{i}}w^{(i)}\cdot(\ell\cdot p^{(i)}+\sum_{i'<i}x_{i'}\cdot p^{(i')}) & = & w^{(i)}\cdot x_{i}\cdot\sum_{i'<i}x_{i'}\cdot p^{(i')}+w^{(i)}\cdot p^{(i)}\sum_{\ell=1}^{x_{i}}\ell\\
 & = & w^{(i)}\cdot x_{i}\cdot\sum_{i'<i}x_{i'}\cdot p^{(i')}+w^{(i)}\cdot p^{(i)}\cdot\frac{x_{i}\cdot(x_{i}+1)}{2}=:s_{i}(x).
\end{eqnarray*}
Note that $s_{i}(x)$ is a convex polynomial of degree 2 (being the
sum of quadratic polynomials with only positive coefficients). Observe
that when scheduling $x_{i}$ jobs of type~$i$, it is optimal to
reject the $n_{i}-x_{i}$ jobs with lowest rejection costs among all
jobs of type $i$. Assume the jobs of each type $i$ are labeled
$j_{1}^{(i)},\hdots,j_{n_{i}}^{(i)}$ by non-decreasing rejection
costs. For each $s\in\mathbb N$ let $f_{i}(s):=\sum_{\ell=1}^{n_{i}-s}e_{j_{\ell}^{(i)}}$.
In particular, to schedule~$x_{i}$ jobs of type~$i$ we
can select them such that we need to pay $f_{i}(s)$ for rejecting
the non-scheduled jobs (and this is an optimal decision). The difficulty
is that the function $f_{i}(s)$ is in general not expressible
by a polynomial whose degree is globally bounded (i.e., for each possible
instance), which prevents a direct application of Theorem~\ref{thm:convexintegerminimizationisfpt}.

However, in Lemma~\ref{thm:polynomialrepresentation} we show that
$f_{i}(s)$ is the maximum of $n_{i}$ linear polynomials, allowing
us to formulate a convex program and solve it by Theorem~\ref{thm:convexintegerminimizationisfpt}. 

\begin{lemma} \label{thm:polynomialrepresentation} For
each type $i$ there is a set of $n_{i}$ polynomials $p_{i}^{(1)},\hdots,p_{i}^{(n_{i})}$
of degree one such that $f_{i}(s)=\max_{\ell}p_{i}^{(\ell)}(s)$ for
each $s\in\{0,\hdots,n_{i}\}$. \end{lemma}
\begin{proof}
For each $\ell\in\{0,\hdots,n_{i}-1\}$ we define $p_{i}^{(\ell)}(s)$
to be the unique polynomial of degree one such that $p_{i}^{(\ell)}(\ell)=f_{i}(\ell)$
and $p_{i}^{(\ell)}(\ell+1)=f_{i}(\ell+1)$. We observe that~$f_{i}(s)$
is convex, since $f_{i}(\ell)-f_{i}(\ell+1)=e_{j_{n_{i}-\ell}^{(i)}}\ge e_{j_{n_{i}-\ell-1}^{(i)}}=f_{i}(\ell+1)-f_{i}(\ell+2)$
for each $\ell\in\{0,\hdots,n_{i}-2\}$. Therefore, the definition of
the polynomials implies that $f_{i}(s)\ge p_{i}^{(\ell)}(s)$ for
each $\ell\in\{0,\hdots,n_{i}-1\}$ and each $s\in\{0,\hdots,n_{i}\}$.
Since for each $s\in\{0,\hdots,n_{i}-1\}$ we have that $p_{i}^{(s)}(s)=f_{i}(s)$
and $p_{i}^{(n_{i}-1)}(n_{i})=f_{i}(n_{i})$, we conclude that $f_{i}(s)=\max_{\ell}p_{i}^{(\ell)}(s)$
for each $s\in\{0,\hdots,n_{i}\}$.\qed \end{proof}

Lemma~\ref{thm:polynomialrepresentation}
allows modeling the entire problem with the following convex program, where
for each type~$i$, variable $g_{i}$ models the rejection costs for
jobs of type~$i$. 
\begin{multline*}
  \min~\sum_{i=1}^{\tau}g_{i}+s_{i}(x)~\textnormal{s.t.}~\sum_{i=1}^{\tau}(n_{i}-x_{i})\leq k,
  g_{i}\geq p_{i}^{(\ell)}(x_{i})~\forall~i\in\{1,\hdots,\tau\}~\forall~\ell\in\{1,\hdots,n_{i}\},
  \mathbf{g},\mathbf{x}\in\mathbb{Z}_{\geq0}^{t} \enspace .
\end{multline*}
Observe that this convex program
admits an optimal solution with $g_{i}=\max_{\ell}p_{i}^{(\ell)}(x_{i})=f_{i}(x_{i})$
for each $i$. Thus, solving the convex program
yields an optimal solution to the overall instance. 

\begin{theorem} 
\label{thm:rejectionwithtypes} For sets of $n$ jobs of $\tau$
types the problem $1||\sum_{\leq k}e_{j}+\sum w_{j}C_{j}$ can be
solved in time $(n+\log (\max_j \max \{e_j, p_j, w_j \}))^{O(1)}\cdot2^{O(\tau^3)}$. \end{theorem}
\begin{proof}
The number of integer variables is $2 \tau$. All constraints and the
objective function can be expressed by convex polynomials of degree
at most 2. Hence, using Theorem~\ref{thm:convexintegerminimizationisfpt}
we obtain a running time of $(1+n\cdot \tau)\cdot \log (\max_j \max\{e_j, p_j, w_j\})^{O(1)} 2^{O(\tau^3)}$.
\qed
\end{proof}

\section{Profit Maximization for General Scheduling Problems}

\label{sec:schedulekjobswithelllengths} 
%
In this section, we consider the parameterized dual problem to scheduling
jobs with rejection: We consider the problem to reject at least $n-s$ jobs
($s$ being the parameter)
to minimize the total cost given by the rejection penalties plus the
cost of the schedule. As we will show, this is equivalent to selecting
at most $s$ jobs and to schedule them in order to maximize the profit of
the scheduled jobs. In contrast to the previous section, here we allow
jobs to have release dates and that the profit of a job depends on
the time when it finishes in the schedule, described by a function
that might be different for each job (similarly as for job dependent
cost functions).

Formally, we are given a set $J$ of $n$ jobs, where each job $j$
is characterized by a release date $r_{j}$, a processing time $p_{j}$,
and a non-increasing profit function~$f_{j}(t)$. We assume that
the functions $f$ are given as oracles and that we can evaluate them
in time $O(1)$ per query. Let $\overline{p}$
denote the number of distinct processing times~$p_{j}$ in the instance.
We want to schedule a set $\bar{J}\subseteq J$ of at most $s$ jobs
from~$J$ on a single machine. Our objective is to maximize $\sum_{j\in\bar{J}}f_{j}(C_{j})$,
where~$C_{j}$ denotes the completion time of~$j$ in the computed
schedule. We call this problem the \emph{$s$-bounded General Profit
Scheduling Problem}, or \emph{$s$-GPSP} for short. Observe that this
generic problem definition allows to model profit functions that stem
from scheduling objectives such as weighted flow time and weighted
tardiness. 

We like to note that in a first version of this paper we presented
an algorithm using three, rather than two parameters. An anonymous
referee gave us advice how to remove one parameter and we present
the improved version here. First, we show how to find the best solution
with a non-preemptive schedule, and subsequently we show how to extend
the algorithm to the preemptive setting.


\subsection{Non-Preemptive Schedules}

\label{sec:profit-max-nonpreemptive} Denote by $q_{1},\hdots,q_{\overline{p}}$
the set of arising processing times in the instance. We start with
a color-coding step, where we color each job uniformly at random with
one color from $\{1,\hdots,s\}$. When coloring the jobs randomly,
with probability~$s!/s^{s}$ all jobs scheduled in an optimal solution
(the set $\bar{J}$) will receive distinct colors. We call such a
coloring of $\bar{J}$ a \emph{good coloring}. When trying enough
random colorings, with high probability we will find at least one good
coloring. We will discuss later how to derandomize this algorithm.
For now, assume a fixed good coloring $c:J\rightarrow\{1,\hdots,s\}$.

We describe now a dynamic program that finds the best non-preemptive
schedule of at most $s$ jobs with different colors. The dynamic program
has one cell $C(S,t)$ for each combination of a set $S\subseteq\{1,\hdots,s\}$
and a time $t\in T$, where $T:=\{r_{j}+\sum_{i=1}^{\overline{p}}s_{i}\cdot q_{i}~|~j\in J,s_{1},\hdots,s_{\overline{p}}\in\{0,\hdots,s\}\}$.
Observe that the set~$T$ contains all possible completion times of
a job in an optimal non-preemptive schedule (we will show later that
this is even true for preemptive schedules). Also, $|T|\le n\cdot(s+1)^{\overline{p}}$.
The cell $C(S,t)$ models the subproblem of scheduling a set of at
most $|S|$ jobs such that their colors are pairwise different and
contained in $S$, and so that the last job finishes by time $t$
the latest. We will denote by~$\mathsf{opt}(S,t)$ the value of the optimal
solution for this subproblem \aw{with respect to the fixed coloring $c$}.

Consider an optimal solution with value $\mathsf{opt}(S,t)$. If no
job finishes at time $t$, then $\mathsf{opt}(S,t)=\mathsf{opt}(S,t')$
for some $t'<t$. Else, there is some job $j$ with $c(j)\in S$ that
finishes at time $t$; then $\mathsf{opt}(S,t)=f_{j}(t)+\mathsf{opt}(S\setminus\{c(j)\},t')$,
where $t'$ is the largest value in $T$ with $t'\le t-p_{j}$. We now
prove this formally. \begin{lemma} \label{lem:DP-trans-non-pmtn}
Let $c:J\rightarrow\{1,\hdots,s\}$ be a coloring, let $S\subseteq\{1,\hdots,s\}$
and let $t\in T$. The optimal value $\mathsf{opt}(S,t)$ for the
dynamic programming cell $C(S,t)$ equals 
\begin{multline}
\max\Big\{ \max_{t'\in T:~t'<t}\mathsf{opt}(S,t'),\max_{\substack{j\in J:~(r_{j}\le t-p_{j})\wedge c(j)\in S}}f_{j}(t)+\mathsf{opt}(S\setminus\{c(j)\},\max\{t'\in T~|~t'\le t-p_{j}\}),0\Big\}.\label{eq:non-pmtn-DP-transition}
\end{multline}
\end{lemma}
\begin{proof} For cells $C(S,t_{\min})$ with $t_{\min}=\min_{t\in T}t$
the claim is true since $t_{\min}=\min_{j}r_{j}$ and so $\mathsf{opt}(S,t_{\min})=0$.
Now consider a cell $C(S,\bar{t})$ with $\bar{t}>t_{\min}$. Let
$\mathsf{opt}'(S,\bar{t})$ denote the right-hand side of~\eqref{eq:non-pmtn-DP-transition}.
We show that $\mathsf{opt}(S,\bar{t})\ge\mathsf{opt}'(S,\bar{t})$
and $\mathsf{opt}(S,\bar{t})\le\mathsf{opt}'(S,\bar{t})$.

For the first claim, we show that there is a schedule with value $\mathsf{opt}'(S,\bar{t})$.
If $\mathsf{opt}'(S,\bar{t})=0$ or $\mathsf{opt}'(S,\bar{t})=\max_{t'\in T:t'<\bar{t}}\mathsf{opt}(S,t')$
then the claim is immediate. Now suppose that there is a job $j$
with $r_{j}\le \bar{t}-p_{j}$ and $c(j)\in S$ such that $\mathsf{opt}(S,\bar{t})=f_{j}(\bar{t})+\mathsf{opt}(S\setminus\{c(j)\},\max\{t'\in T|t'\le \bar{t}-p_{j}\})$.
Then there is a feasible schedule for the problem encoded in cell
$C(S,\bar{t})$ that is given by the jobs in the optimal solution
for the cell $C(S\setminus\{c(j)\},\max\{t'\in T~|~t'\le\bar{t}-p_{j}\})$,
and additionally schedules job $j$ during $[\bar{t}-p_{j},\bar{t})$.
The profit of this solution is $f_{j}(\bar{t})+\mathsf{opt}(S\setminus\{c(j)\},\max\{t'\in T~|~t'\le \bar{t}-p_{j}\})$.

For showing that $\mathsf{opt}(S,\bar{t})\le\mathsf{opt}'(S,\bar{t})$
consider the schedule for $\mathsf{opt}(S,\bar{t})$ and assume w.l.o.g.~that 
the set $T$ contains all start and end times of jobs in that schedule. If no job finishes
at time $\bar{t}$ then $\mathsf{opt}(S,\bar{t})=\max_{t'\in T:~t'<t}\mathsf{opt}(S,t')$,
using that $T$ contains all finishing times of a job in the optimal schedule. Note that $\{t'\in T|t'<t\}\ne\emptyset$
since we assumed that $\bar{t}>t_{\min}$. If a
job $j$ finishes at time $\bar{t}$ then $\mathsf{opt}(S,\bar{t})=f_{j}(\bar{t})+\mathsf{opt}(S\setminus\{c(j)\},\max\{t'\in T~|~t'\le \bar{t}-p_{j}\})$.
This implies that $\mathsf{opt}(S,\bar{t})\le\mathsf{opt}'(S,\bar{t})$.
\qed
\end{proof}

To compute the value $\mathsf{opt}(S,t)$, our dynamic program
evaluates expression~\eqref{eq:non-pmtn-DP-transition}. Note that
this transition implies $\mathsf{opt}(S,t_{\min})=0$ for $t_{\min}$
the minimum value in~$T$. The number of dynamic programming
cells is at most $2^{s}\cdot|T|=2^{s} n\cdot(s+1)^{\overline{p}}$.
Evaluating expression~\eqref{eq:non-pmtn-DP-transition} takes time
$O(|T|+n)=O((s+1)^{\overline{p}}+n)$, given that optimal solutions
to all subproblems are known. Together with Lemma~\ref{lem:DP-trans-non-pmtn}
this yields the following lemma.

\begin{lemma} There is an algorithm that, given an instance of $s$-GPSP
with a set~$J$ of~$n$ jobs of at most~$\overline{p}$ distinct
processing times together with a good coloring $c:J\rightarrow\{1,\hdots,s\}$,
in time $O(2^{s}\cdot n^{2}\cdot(s+1)^{2\overline{p}})$ computes an optimal non-preemptive schedule.
\end{lemma} Instead of coloring the jobs randomly, we can use a family
of hash functions, as described by Alon et al.~\cite{AlonEtAl1995}.
Using our notation, they show that there is a family~$\mathcal{F}$
of $2^{O(s)}\log n$ hash functions $J\rightarrow\{1,\hdots,s\}$
such that for any set $J'\subseteq J$ (in particular for the set
of jobs $J'$ that is scheduled in an optimal solution) there is a
hash function $F\in\mathcal{F}$ such that $F$ is bijective on $J'$.
The value of each of these functions on each specific element of $J$
can be evaluated in $O(1)$ time. This yields the following. \begin{theorem}
\label{thm:reject-n-k-non-pmtn} There is a deterministic algorithm
that, given an instance of \textsc{$s$-GPSP} with~$n$ jobs and $\overline{p}$
processing times, in time
$2^{O(s)}\cdot(s+1)^{2\overline{p}}\cdot n^{2}\log n$ computes an optimal non-preemptive schedule. \end{theorem}

\subsection{Preemptive Schedules}

We now extend our algorithm from the previous section to the setting
when preemption of jobs is allowed. Now the dynamic program becomes
more complicated. As before, we color all jobs in $J$ uniformly at
random with colors from $\{1,\hdots,s\}$. Fix a good coloring $c:J\rightarrow\{1,\hdots,s\}$.
In the dynamic programming table we have one cell $C(S,t_{1},t_{2},P)$
for each combination of a set $S\subseteq\{1,\hdots,s\}$, values
$t_{1},t_{2}\in T$, where as above $T=\{r_{j}+\sum_{i=1}^{\overline{p}}s_{i}\cdot q_{i}~|~j\in J,s_{1},\hdots,s_{\overline{p}}\in\{0,\hdots,s\}\}$,
and a value $P\in\P:=\{\sum_{i=1}^{\overline{p}}s_{i}\cdot q_{i}~|~s_{1},\hdots,s_{\overline{p}}\in\{0,\hdots,s\}\}$
(note that like for $T$ we have $|\P|\le(s+1)^{\overline{p}}$).
This cell encodes the problem of selecting at most~$|S|$ jobs with total processing time at most $P$ such
that 
(i)~their colors are pairwise different and contained in~$S$, 
(ii)~they are released during $[t_{1},t_{2})$;
we want to schedule them during $[t_{1},t_{2})$, possibly with
preemption, the objective being to maximize the total profit.

To this end, we first prove that the set $T$ contains all possible
start- and completion times of jobs of an optimal solution. The former
denotes the first time in a schedule where each job is being processed.
\begin{lemma}\label{lem:start-completion-times-T} There is an optimal
solution such that for each scheduled job $j$, its start time and
completion time are both contained in $T$. \end{lemma} 
\begin{proof}
Consider an optimal schedule $\mathcal{S}$ for the set of at most~$k$ jobs that we want to schedule. If we fix the completion times
of the jobs in the schedule as deadlines and run the policy of Earliest-Deadline-First
(EDF) with these deadlines, then we obtain a schedule $\mathcal{S}'$
where each job finishes no later than in~$\mathcal{S}$. This is
true since given one machine and jobs with release dates and deadlines,
there is a feasible preemptive schedule (i.e., a schedule where each
job finishes by its deadline) if and only if EDF finds such a schedule.
In particular, since the profit-functions of
the jobs are non-increasing, the profit given by $\mathcal{S}'$ is
at most the profit given by $\mathcal{S}$.

We prove the claim by induction over the deadlines of the jobs in
$\mathcal{S}'$. In fact, we prove the following slightly stronger
claim: for the job with $s'$-th deadline, its start time and completion
time are in the set $T_{s'}=\{r_{j}+\sum_{i=1}^{\overline{p}}s_{i}\cdot q_{i}~|~j\in J,s_{1},\hdots,s_{\overline{p}}\in\{0,\hdots,s'\}\}$.

For the job $j_{1}$ with smallest deadline, we have that $j_{1}$
starts at time $r_{j_{1}}\in T_{1}$ and finishes at time $r_{j_{1}}+p_{j_{1}}\in T_{1}$,
since $j_{1}$ has highest priority. Suppose by induction that the
claim is true for the jobs with the $s'-1$ smallest deadlines, and
note that the deadlines are distinct. Consider the job $j_{s'}$ with
the $s'$-smallest deadline. Let $R_{s'}$ denote its start time.
Let $J_{1}$ denote the jobs finishing before~$R_{s'}$. Similarly,
let $J_{2}$ denote the jobs finishing during $(R_{s'}C_{s'}]$. Note
that $j_{s'}\in J_{2}$ but $J_{2}$ does not contain a job with deadline
later than $j_{s'}$. Since we run EDF, $R_{s'}=r_{j_{s'}}$ or $R_{s'}=C_{j_{s''}}$
for some job $j_{s''}$. In both cases $R_{s'}\in T_{|J_{1}|}$, and
thus $R_{s'}=r_{j}+\sum_{i=1}^{\overline{p}}s_{i}^{(1)}\cdot q_{i}$
for some job $j\in J$ and values $s_{1}^{(1)},\hdots,s_{\overline{p}}^{(1)}\in\{0,\hdots,|J_{1}|\}$.
During $(R_{s'}C_{s'}]$, at most~$|J_{2}|$ jobs finish, and their
total length is $\sum_{i=1}^{\overline{p}}s_{i}^{(2)}\cdot q_{i}$
for values $s_{1}^{(2)},\hdots,s_{\overline{p}}^{(2)}\in\{0,\hdots,|J_{2}|\}$.
Hence, 
\[
C_{s'}=r_{j}+\sum_{i=1}^{\overline{p}}s_{i}^{(1)}\cdot q_{i}+\sum_{i=1}^{\overline{p}}s_{i}^{(2)}\cdot q_{i}=r_{j}+\sum_{i=1}^{\overline{p}}s_{i}\cdot q_{i}
\]
for some job $j\in J$ and values $s_{1},\hdots,s_{\overline{p}}\in\{0,\hdots,|J_{1}|+|J_{2}|\}$.
By $|J_{1}|+|J_{2}|=s'$, we have that $C_{j_{s'}}\in T_{s'}$. 
\qed
\end{proof}

Suppose we want to compute an optimal solution of value $\mathsf{opt}(S,t_{1},t_{2},P)$
for a cell $C(S,t_{1},t_{2},P)$. In any optimal solution, we can
assume that the last finishing job has smallest priority, i.e., we
work on this job only if no other job is pending. Let $j$ be this
(last) job and let $R_{j}$ denote its start time. Then, during~$[R_{j},t_{2})$
the only scheduled jobs are $j$ and some jobs that are released during~$[R_{j},t_{2})$. Also, all jobs scheduled during~$[t_{1},R_{j})$
are also released during~$[t_{1},R_{j})$ and they are finished within
$[t_{1},R_{j})$.

This observation gives rise to the dynamic programming transition.
Given a cell $C(S,t_{1},t_{2},P)$, we enumerate all possibilities
for the last job $j$, its start time~$R_{j}$, the total processing
time of jobs scheduled during $[t_{1},R_{j})$ and $[R_{j},t_{2})$,
and their colors. One combination corresponds to an optimal solution,
and we reduce the overall problem to the respective two subproblems
for $[t_{1},R_{j})$ and~$[R_{j},t_{2})$. (Observe that the color
coding step prevents to select the same job in both subproblems.) We justify
this transition in the following lemma. For ease of notation, for
a cell $C(S,t_{1},t_{2},P)$ denote by $\mathsf{en}(S,t_{1},t_{2},P)$
the set of tuples $(j,S',R_{j},P',P'')$ that we need to enumerate
for splitting the cell, i.e., all such tuples with 
$c(j)\in S$, 
$t_{1}\le r_{j}\le t_{2}-p_{j}$, 
$S'\subseteq S\setminus\{c(j)\}$, 
$t_{2}-R_{j}\ge P''+p_{j}$, 
$P',P''\in\P$, and 
$P'+P''+p_{j}\le P$. 
\begin{lemma} \label{lem:non-pmtn-DP-transition}For a coloring $c:J\rightarrow\{1,\hdots,s\}$,
each dynamic programming cell $C(S,t_{1},t_{2},P)$ satisfies
\begin{multline}
\mathsf{opt}(S,t_{1},t_{2},P)=\max\Big\{\max_{\substack{t'\in T:~t_{1}\le t'<t_2}}\mathsf{opt}(S,t_{1},t',P),\\
\max_{\substack{(j,S',R_{j},P',P'')\in\mathsf{en}(S,t_{1},t_{2},P)}}f_{j}(t_{2})+\mathsf{opt}(S',t_{1},R_{j},P')+\mathsf{opt}(S\setminus(S'\cup\{c(j)\},R_{j},t_{2},P''),0\Big\} \enspace . \label{eq:pmtn-DP-transition}
\end{multline}
\end{lemma}
\begin{proof}
Similarly as in the proof of Lemma~\ref{lem:DP-trans-non-pmtn},
let $\mathsf{opt}'(S,t_{1},t_{2},P)$ denote the right-hand side of
\eqref{eq:pmtn-DP-transition} for the cell $C(S,t_{1},t_{2},p)$.

Consider a cell $C(S,t_{1},t_{2},P)$. First, we show that $\mathsf{opt}(S,t_{1},t_{2},P)\ge\linebreak\mathsf{opt}'(S,t_{1},t_{2},P)$.
To this end, it suffices to show that there is a feasible solution
with value at least $\mathsf{opt}'(S,t_{1},t_{2},P)$. If $\mathsf{opt}'(S,t_{1},t_{2},P) = 0$
then the claim is immediate (just take the empty solution). Otherwise,
first assume that $\mathsf{opt}'(S,t_{1},t_{2},P)=\mathsf{opt}(S,t_{1},t',P)$
for some $t'\in T$ with $t'<t_2$. Any solution for the cell $C(S,t_{1},t',P)$
is also feasible for $C(S,t_{1},t_{2},P)$, and thus $\mathsf{opt}(S,t_{1},t_{2},P)\ge\mathsf{opt}(S,t_{1},t',P)$.
Hence, $\mathsf{opt}(S,t_{1},t_{2},P)\ge\mathsf{opt}'(S,t_{1},t_{2},P)$.

Finally, assume that there is a tuple $(j,S',R_{j},P',P'')$ such
that
$$\mathsf{opt}'(S,t_{1},t_{2},P)=f_{j}(t_{2})+\mathsf{opt}(S',t_{1},R_{j},P')+\mathsf{opt}(S\setminus(S'\cup\{c(j)\},R_{j},t_{2},P'')\enspace .$$
We build a solution with this value as follows. During $[t_{1},R_{j})$,
we schedule an optimal solution for the cell $C(S',t_{1},R_{j},P')$.
During $[R_{j},t_{2})$, we schedule an optimal solution for the cell
$C(S\setminus(S'\cup\{c(j)\},R_{j},t_{2},P'')$. By assumption, $t_{2}-R_{j}\ge P''+p_{j}$.
Hence, the schedule for the latter cell leaves total idle time during
$[R_{j},t_{2})$ of at least~$p_{j}$. During these idle times we
schedule job $j$ and, thus, $j$ finishes at time $t_{2}$ the latest.
As $S'\subseteq S\setminus\{c(j)\}$ and $S'\cap(S\setminus(S'\cup\{c(j)\})=\emptyset$,
no job is scheduled in both subproblems. Also, $P'+P''+p_{j}\le P$.
Hence, the solution is feasible for the cell $C(S,t_{1},t_{2},P)$
and its value is at least $f_{j}(t_{2})+\mathsf{opt}(S',t_{1},R_{j},P')+\mathsf{opt}(S\setminus(S'\cup\{c(j)\},R_{j},t_{2},P'')$.
Thus, $\mathsf{opt}(S,t_{1},t_{2},P)\ge\mathsf{opt}'(S,t_{1},t_{2},P)$.

Conversely, we want to show that $\mathsf{opt}(S,t_{1},t_{2},P)\le\mathsf{opt}'(S,t_{1},t_{2},P)$.
Consider an optimal solution for the cell $C(S,t_{1},t_{2},P)$. If
$\mathsf{opt}(S,t_{1},t_{2},P)=0$, then the claim is immediate. Otherwise,
we can assume that in the schedule for the solution, the jobs are
ordered according to EDF by their respective completion time (see
also the proof of Lemma~\ref{lem:start-completion-times-T}). Let
$j$ be the job with largest deadline (and hence smallest priority);
then $c(j)\in S$. If $j$ ends at a time $t'<t_{2}$ then $\mathsf{opt}(S,t_{1},t_{2},P)=\mathsf{opt}(S,t_{1},t',P)\le\mathsf{opt}'(S,t_{1},t_{2},P)$.

Now suppose that $j$ finishes at time $t_{2}$, and let $R_{j}$
denote its start time in an optimal schedule of value $\mathsf{opt}(S,t_{1},t_{2},P)$.
Observe that $t_{1}\le r_{j}\le t_{2}-p_{j}$. Let $S'\subseteq S\setminus\{c(j)\}$
denote the set of colors of jobs scheduled during $[t_{1},R_{j})$
in an optimal solution with value $\mathsf{opt}(S,t_{1},t_{2},P)$ (observe that $c(j)\notin S'$), and let~$P'$ denote the total processing time of jobs scheduled during $[t_{1},R_{j})$.
Since all jobs scheduled during $[t_{1},R_{j})$ start and finish
during that interval, we have that $P'\in\P$. Also, let $P''$ denote
the total processing time of jobs other than $j$ that are scheduled
during $[R_{j},t_{2})$ in a solution of value $\mathsf{opt}(S,t_{1},t_{2},P)$.
In particular, $t_{2}-R_{j}\ge P''+p_{j}$. As all those jobs start
and finish during $[R_{j},t_{2})$, this implies that $P''\in\P$.
Hence, $(j,S',R_{j},P',P'')\in\mathsf{en}(S,t_{1},t_{2},P)$, and
thus 
\begin{multline*}
\mathsf{opt}(S,t_{1},t_{2},P)=f_{j}(t_{2})+\mathsf{opt}(S',t_{1},R_{j},P')+\mathsf{opt}(S\setminus(S'\cup\{c(j)\},R_{j},t_{2},P'')\le\mathsf{opt}'(S,t_{1},t_{2},P).
\end{multline*}
\qed
\end{proof}

Our dynamic program evaluates equation~\eqref{eq:pmtn-DP-transition}
for each cell. Observe that cells $C(S,t_{1},t_{2},p)$ with $S=\emptyset$
or $t_{1}=t_{2}$ obey $\mathsf{en}(S,t_{1},t_{2},P)=\emptyset$, and
so $\mathsf{opt}(S,t_{1},t_{2},p)=0$. This yields the following. \begin{lemma} There is an algorithm that, given a set
of $n$ jobs with $\overline{p}$ distinct processing times and
a good coloring $c:J\rightarrow\{1,\hdots,s\}$, in time $2^{O(s)}s^{O(\overline{p})}n^{4}$
computes an optimal preemptive schedule for $s$-GSP. \end{lemma}
Using the same derandomization technique as above, we obtain the following
theorem. \begin{theorem} \label{thm:reject-n-k-pmtn} There is a
deterministic algorithm that, given a set of $n$ jobs with~$\overline{p}$
distinct processing times, in time $2^{O(s)}s^{O(\overline{p})}n^{4}\log n$
computes an optimal preemptive schedule for $s$-GSP. \end{theorem}

Observe that $(\log n)^{O(\overline{p})} = O((\overline{p}\log\overline{p})^{O(\overline{p})} + n)$; this yields the following.
\begin{corollary}
  There is a deterministic algorithm that, given a set of $n$ jobs with~$\overline{p}$ distinct processing times, computes an optimal preemptive schedule for~$O(\log n)$ jobs in time $(\overline{p}\log\overline{p})^{O(\overline{p})}\cdot n^{O(1)}$.
\end{corollary}

Applying Theorem~\ref{thm:reject-n-k-pmtn}, we obtain the following corollary when at least $n-s$ jobs have to
be rejected. 

\begin{corollary}\label{cor:scheuduling-rejection} There is a deterministic
algorithm that, given a set of $n$ jobs with at most
$\overline{p}$ distinct processing times, each job having a private cost function~$w_{j}$, solves the problem $1|r_{j},(pmtn)|\sum_{S}w_{j}(C_{j})+\sum_{\bar{S}}e_{j}$
in time\linebreak $2^{O(s)}s^{O(\overline{p})}n^{4}\log n$, when at least $n-s$
jobs have to be rejected. \end{corollary} 
\begin{proof}Any instance
of the problem to reject $n-s$ jobs and schedule the remaining $s$
jobs to minimize $1|r_{j},(pmtn)|\sum_{S}w_{j}(C_{j})+\sum_{\bar{S}}e_{j}$
is equivalent to an instance of $s$-GPSP (with or without preemption, respectively) where for each job $j$
we define a profit function $f_{j}(t):=e_{j}-w_{j}(t)$ and keep the
release dates and processing times unchanged. Then the corollary follows
from Theorem~\ref{thm:reject-n-k-non-pmtn} and Theorem~\ref{thm:reject-n-k-pmtn}.
\qed
\end{proof}

Observe that the objective function $\sum_{S}f_{j}(C_{j})+\sum_{\bar{S}}e_{j}$
is in particular a generalization of weighted flow time with rejection ($\sum_{S}w_{j}(C_{j}-r_{j})+\sum_{\bar{S}}e_{j}$ 
where each job $j$ has a weight $w_{j}$ associated with it), weighted
tardiness with rejection ($\sum_{S}w_{j}\max\{0,t-d_{j}\}+\sum_{\bar{S}}e_{j}$ and
each job $j$ has additionally a deadline $d_{j}$), and weighted
sum of completion times with rejection ($\sum_{S}w_{j}C_{j}+\sum_{\bar{S}}e_{j}$). So our algorithm works
for these objective functions as well.

When only the number~$s$ of scheduled jobs is chosen as parameter 
the problem becomes $\mathsf{W}[1]$-hard, as pointed out to us by an anonymous reviewer.


\begin{theorem}\label{thm:gsponlynumscheduledjobshardness} Scheduling
a set $\overline{J}$ of $s$ jobs from among $n$ given jobs (non-) preemptively
on a single machine to maximize $\sum_{j\in\bar{J}}b_{j}-w_{j}(C_{j})$
is $\mathsf{W}[1]$-hard for parameter $s$, even if all $n$ jobs
have the same function $w_{j}$ and $b_{j}=p_j$ for each job $j$. 
Thus, (non-)preemptive $s$-GPSP is $\mathsf{W}[1]$-hard for parameter~$s$. \end{theorem}
\begin{proof}
Consider an instance of $k$-\textsc{Subset Sum}, specified by integers
$s_{i}$ and a target value~$q$. In our reduction, for each $s_{j}$,
we create a job~$j$ with processing time $p_{j}:=s_{j}$, profit
$b_{j}:=s_{j}$ and cost function $w_{j}$ defined as 
$$w_{j}(t)=\begin{cases}
0, & ~\mbox{if}~t\leq q,\\
+\infty, & ~\mbox{if}~t>q,
\end{cases}$$ 
for all times $t$. 
Then there is a set of $s$ integers~$s_{i}$ whose sum is exactly
$q$, if and only if we can schedule the corresponding set~$\bar{J}$
of~$s$ jobs preemptively or non-preemptively on a single machine such that 
$\sum_{j\in\bar{J}}b_{j}-w_{j}(C_{j})=q$.
Notice that there is only a single function $w_{j}(t)$ that is used for all
jobs $j\in J$ in this reduction. When defining as profit function
$f_{j}(t):=b_{j}-w_{j}(t)$ for each job $j$ this yields an instance
of $k$-GPSP which is thus $\mathsf{W}[1]$-hard.
\qed
\end{proof}


\subsection{Hardness for Constantly Many Processing Times}

On the other hand, we prove that choosing only the maximum processing time $p_{\max}$, or number of distinct processing times $\bar{p}$ as a parameter
is not enough, as we show the problem to be $\mathsf{NP}$-hard even if $p_j\in \{1,3\}$ for all jobs~$j$.
The same holds for the General Scheduling Problem (GSP)~\cite{BansalPruhs2010} where---without the profit maximization aspect---we are given
a set of jobs~$j$, each of them having a release date $r_j$ and an individual cost function~$w_j(t)$,
and we want to schedule all of them on one machine in order to minimize $\sum_j w_j(C_j)$ where $C_j$ denotes the
completion time of job $j$ in the computed (preemptive) schedule.

\begin{theorem}
\label{thm:GPSP-hardness}
The General Profit Scheduling Problem (GPSP) and the General Scheduling Problem (GSP) are $\mathsf{NP}$-hard,
even if the processing time of each job is exactly 1 or 3.
This holds in both the preemptive and non-preemptive setting.
\end{theorem}

We now show that GPSP is $\mathsf{NP}$-hard, even if $p_{j}\in\{1,3\}$
for each job $j$. We reduce from {\sc 3-Dimensional Matching}, where one is given
three sets $A=\{a_{1},\hdots,a_{n}\}$, $B=\{b_{1},\hdots,b_{n}\}$, and
$C=\{c_{1},\hdots,c_{n}\}$ for some $n\in\mathbb N$, and a set $T\subseteq A\times B\times C$.
The goal is to find a subset $T'\subseteq T$ with $|T'|=n$ such
that for any two different triples $(a_{i},b_{j},c_{k}),(a'_{i},b'_{j},c'_{k})\in T'$
it holds that $a_{i}\ne a'_{i}$, $b_{j}\ne b'_{j}$, and $c_{k}\ne c'_{k}$.
Note that this implies that each element in $A$, $B$, and $C$ appears
in exactly one triple in $T'$. We say that the given instance is
a ``yes''-instance if there exists such a set $T'$. The {\sc 3-Dimensional Matching} problem is $\mathsf{NP}$-hard~\cite{GareyJohnson1979}.

Given an instance of {\sc 3-Dimensional Matching}, we construct an instance
of GPSP. With each triple $(a_{i},b_{j},c_{k})\in T$ we associate
an interval of length 3 on the time axis. For each such triple, we
define a value $t(a_{i},b_{j},c_{k})$ such that for any two different
triples, the corresponding intervals $[t(a_{i},b_{j},c_{k}),t(a_{i},b_{j},c_{k})+3)$
are disjoint and $\bigcup_{(a_{i},b_{j},c_{k})\in T}[t(a_{i},b_{j},c_{k}),t(a_{i},b_{j},c_{k})+3)=[0,3|T|)$.
This can for instance be achieved by assigning to each triple $(a_{i},b_{j},c_{k})\in T$
a unique identifier $z_{(a_{i},b_{j},c_{k})}\in\{0,|T|-1\}$ and defining
$t(a_{i},b_{j},c_{k}):=3\cdot z_{(a_{i},b_{j},c_{k})}$. For each
triple $(a_{i},b_{j},c_{k})\in T$, we introduce four jobs, all
released at time $t(a_{i},b_{j},c_{k})$: three jobs of unit length called $A(a_{i},b_{j},c_{k})$, $B(a_{i},b_{j},c_{k})$, $C(a_{i},b_{j},c_{k})$ that
represent the occurrence of the respective element from $A$, $B$, $C$ in the triple $(a_{i},b_{j},c_{k})$ (so the occurrence of the elements $a_{i}$, $b_{j}$, and $c_{k}$
in $(a_{i},b_{j},c_{k})$, respectively); and additionally one job of length three called $L(a_{i},b_{j},c_{k})$. 

For each element in $A\cup B\cup C$ we specify a unique point in
time. We define $t(a_{i}):=3|T|+i$, $t(b_{j}):=3|T|+n+j$, and $t(c_{k}):=3|T|+2n+k$
for each respective element $a_{i}\in A$, $b_{j}\in B$, $c_{k}\in C$; see Figure~\ref{fig:paranphardness} for a sketch of the subdivision of the time axis and the profit functions (that we define later).
\begin{figure}

\begin{centering}
\includegraphics[scale=0.90]{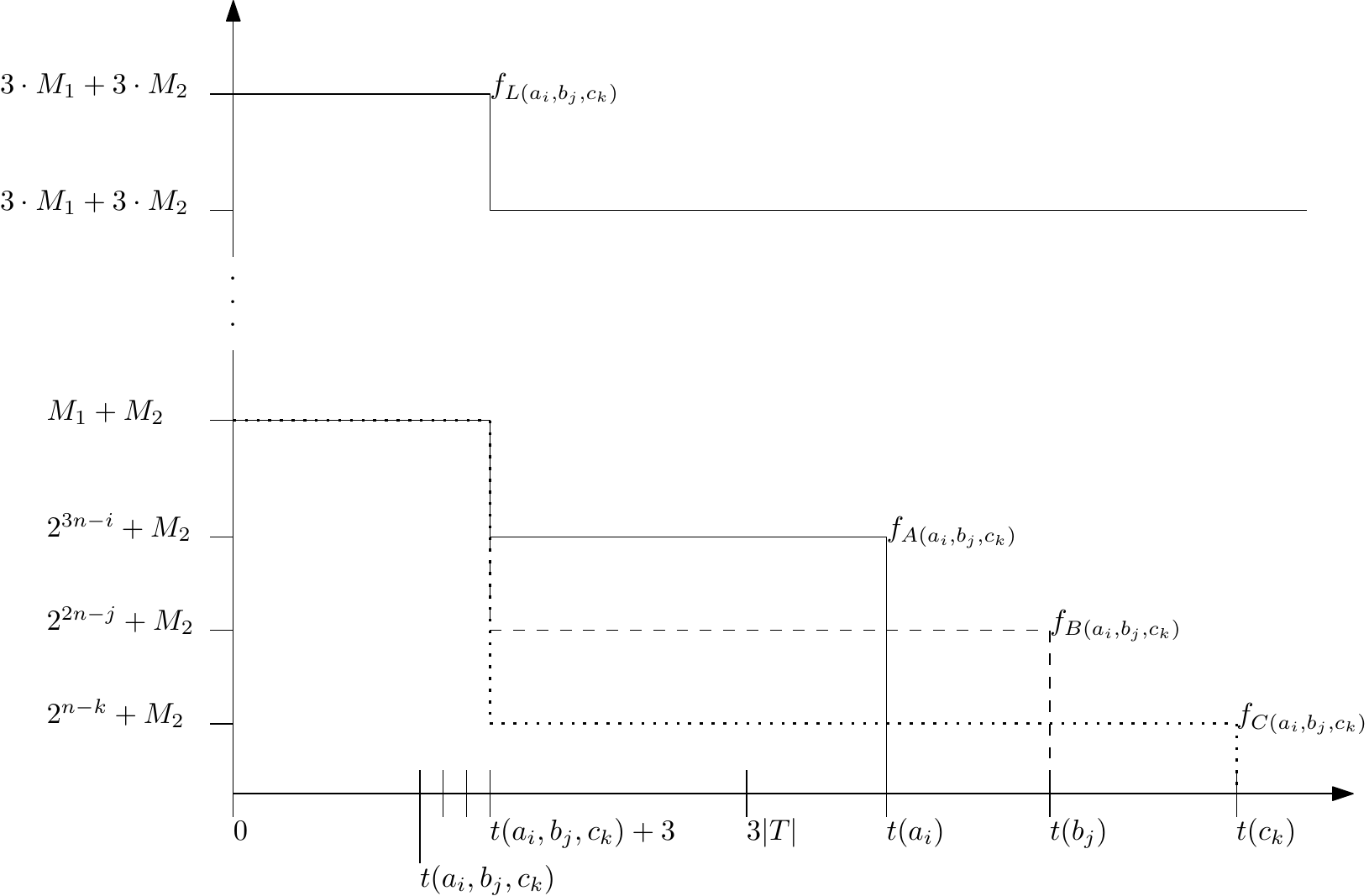}
\par\end{centering}
\caption{\label{fig:paranphardness}Sketch for the reduction in Theorem~\ref{thm:GPSP-hardness}.}
\end{figure}

Before specifying the profit function for each job, we give some intuition.
The idea is that if the given instance is a ``yes''-instance, then there
is an optimal schedule where
\begin{itemize}
\item for each triple $(a_{i},b_{j},c_{k})$ in the optimal solution we
schedule the job $L(a_{i},b_{j},c_{k})$ during the interval $[t(a_{i},b_{j},c_{k}),t(a_{i},b_{j},c_{k})+3)$
and the jobs $A(a_{i},b_{j},c_{k})$, $B(a_{i},b_{j},c_{k})$, and
$C(a_{i},b_{j},c_{k})$ during interval $[t(a_{i})-1,t(a_{i}))$, $[t(b_{j})-1,t(b_{j}))$,
and $[t(c_{k})-1,t(c_{k}))$, respectively.
\item for each triple $(a_{i},b_{j},c_{k})$ that is \emph{not }in the optimal
solution we schedule jobs $A(a_{i},b_{j},c_{k})$, $B(a_{i},b_{j},c_{k})$, $C(a_{i},b_{j},c_{k})$ during~$[t(a_{i},b_{j},c_{k}),t(a_{i},b_{j},c_{k})+3)$
(in an arbitrary order) and the job $L(a_{i},b_{j},c_{k})$ somewhere
after time $3|T|+3n$.
\end{itemize}
For defining the profit functions, let $M_{1}:=1+2^{3n}$ and $M_{2}:=M_{1}\cdot3|T|+2^{3n}$.
For the jobs $A(a_{i},b_{j},c_{k})$, $B(a_{i},b_{j},c_{k})$, and
$C(a_{i},b_{j},c_{k})$ for a triple $(a_{i},b_{j},c_{k})\in T$ we
define their profit function as follows: 

\[
f_{A(a_{i},b_{j},c_{k})}(t):=\begin{cases}
M_{1}+M_{2}, & \mathrm{for}~ t\le t(a_{i},b_{j},c_{k})+3,\\
2^{3n-i}+M_{2}, & \mathrm{for}~ t(a_{i},b_{j},c_{k})+3 < t\le t(a_{i}),\\
0, & \mathrm{for}~ t>t(a_{i});
\end{cases}
\]

\[
f_{B(a_{i},b_{j},c_{k})}(t):=\begin{cases}
M_{1}+M_{2}, & \mathrm{for}~ t\le t(a_{i},b_{j},c_{k})+3,\\
2^{2n-j}+M_{2}, & \mathrm{for}~ t(a_{i},b_{j},c_{k})+3 < t\le t(b_{j}),\\
0, & \mathrm{for}~ t>t(b_{j});
\end{cases}
\]

\[
f_{C(a_{i},b_{j},c_{k})}(t):=\begin{cases}
M_{1}+M_{2}, & \mathrm{for}~ t\le t(a_{i},b_{j},c_{k})+3,\\
2^{n-k}+M_{2}, & \mathrm{for}~ t(a_{i},b_{j},c_{k})+3 < t\le t(c_{k}),\\
0, & \mathrm{for}~ t>t(c_{k}) \enspace .
\end{cases}
\]

For the long job $L(a_{i},b_{j},c_{k})$ for a triple $(a_{i},b_{j},c_{k})\in T$
we define its profit function as

\[
f_{L(a_{i},b_{j},c_{k})}(t):=\begin{cases}
3\cdot M_{1}+3\cdot M_{2}, & \mathrm{for}~ t\le t(a_{i},b_{j},c_{k})+3,\\
3\cdot M_{2}, & \mathrm{for}~ t>t(a_{i},b_{j},c_{k})+3 \enspace .
\end{cases}
\]

To show the correctness of our reduction, we show that there is a
solution with overall profit of $|T|\cdot(3\cdot M_{1}+3\cdot M_{2})+|T|\cdot3\cdot M_{2}+2^{3n}-1$
if and only if the instance of {\sc 3-Dimensional Matching} is a ``yes''-instance. 
\begin{lemma}
\label{lem:yes=>high-profit}If the given instance of {\sc 3-Dimensional Matching}
is a ``yes''-instance, then there is a solution with profit $|T|\cdot(3\cdot M_{1}+3\cdot M_{2})+|T|\cdot3\cdot M_{2}+2^{3n}-1$.\end{lemma}
\begin{proof}
Let $T'\subseteq T$ be a solution of the instance of {\sc 3-Dimensional Matching}
(note that $|T'|=n$). For each tuple $(a_{i},b_{j},c_{k})\in T'$
we schedule the job $L(a_{i},b_{j},c_{k})$ during $[t(a_{i},b_{j},c_{k}),t(a_{i},b_{j},c_{k})+3)$,
and the jobs $A(a_{i},b_{j},c_{k})$, $B(a_{i},b_{j},c_{k}),C(a_{i},b_{j},c_{k})$ during $[t(a_{i})-1,t(a_{i}))$, $[t(b_{j})-1,t(b_{j}))$,
and\linebreak $[t(c_{k})-1,t(c_{k}))$, respectively. Note that since $T'$
is a feasible solution, at most one job is scheduled in each interval
$[t(a_{i})-1,t(a_{i}))$, $[t(b_{j})-1,t(b_{j})),[t(c_{k})-1,t(c_{k}))$.
For all tuples $(a_{i},b_{j},c_{k})\notin T'$ we schedule the job
$L(a_{i},b_{j},c_{k})$ in some arbitrary interval during $[3|T|+3n,\infty)$
and the jobs $A(a_{i},b_{j},c_{k})$, $B(a_{i},b_{j},c_{k})$, and
$C(a_{i},b_{j},c_{k})$ during $[t(a_{i},b_{j},c_{k}),t(a_{i},b_{j},c_{k})+3)$.
This solution yields an overall profit of $3|T|\cdot(M_{1}+M_{2})+\sum_{\ell=0}^{3n-1}(M_{2}+2^{\ell})+(|T|-n)\cdot3\cdot M_{2}=|T|\cdot(3\cdot M_{1}+3\cdot M_{2})+|T|\cdot3\cdot M_{2}+2^{3n}-1$.
\qed
\end{proof}
For the converse of Lemma~\ref{lem:yes=>high-profit}, we first
show that there is always an optimal solution to the defined GPSP instance
that is in standard form. We say that a solution is in \emph{standard
form} if 
\begin{enumerate}
\item each job $A(a_{i},b_{j},c_{k})$ (job $B(a_{i},b_{j},c_{k})$, job $C(a_{i},b_{j},c_{k})$)
is scheduled either during interval \mbox{$[t(a_{i},b_{j},c_{k}),t(a_{i},b_{j},c_{k})+3)$}
or during $[3|T|,t(a_{i}))$ (during $[3|T|,t(b_{j})$, during $[3|T|,t(c_{k})$);
and
\item each job $L(a_{i},b_{j},c_{k})$ is scheduled non-preemptively either
completely during interval \mbox{$[t(a_{i},b_{j},c_{k}),t(a_{i},b_{j},c_{k})+3)$}
or completely during $[3|T|+3n,\infty)$.
\end{enumerate}
Observe in particular that solutions in standard form schedule \emph{all
}jobs.
\begin{lemma}
\label{lem:optimal-solution-standard-form}There is always an optimal
solution of the instance of GPSP that is in standard form.\end{lemma}
\begin{proof}
Given an optimal solution to the defined GPSP instance, suppose that
it does not satisfy the first property. First, suppose for contradiction
that there is a job $A(a_{i},b_{j},c_{k})$ that finishes after time
$t(a_{i})$, thus contributing zero towards the objective. Then the
objective value can be at most $3|T|\cdot(M_{1}+M_{2})+|T|\cdot(3\cdot M_{1}+3\cdot M_{2})-M_{1}-M_{2}$
(assuming that all other jobs give the maximal possible profit). However,
then there is a strictly better solution $S'$ which is given by 
\begin{itemize}
\item scheduling each job $A(a_{i},b_{j},c_{k})$, $B(a_{i},b_{j},c_{k})$,
and $C(a_{i},b_{j},c_{k})$ for each triple $(a_{i},b_{j},c_{k})\in T$
during $[t(a_{i},b_{j},c_{k}),t(a_{i},b_{j},c_{k})+3)$, respectively,
and
\item scheduling each job $L(a_{i},b_{j},c_{k})$ for each triple $(a_{i},b_{j},c_{k})\in T$
somewhere during $[3|T|+3n,\infty)$.
\end{itemize}
The solution $S'$ yields a profit of at least $3|T|\cdot(M_{1}+M_{2})+|T|\cdot(3\cdot M_{2})>3|T|\cdot(M_{1}+M_{2})+|T|\cdot(3\cdot M_{1}+3\cdot M_{2})-M_{2}$,
using that $M_{2}>3|T|\cdot M_{1}$. Hence, the first considered schedule
was not optimal. For jobs $B(a_{i},b_{j},c_{k})$ and $C(a_{i},b_{j},c_{k})$,
the claim can be shown similarly. 

Secondly, for contradiction suppose that there is a job $A(a_{i},b_{j},c_{k})$
that is scheduled during interval $[t(a_{i},b_{j},c_{k})+3,3\cdot|T|)$. Then
the computed schedule yields a profit of at most $3|T|\cdot(M_{1}+M_{2})-M_{1}+\sum_{\ell=0}^{3n-1}2^{\ell}+|T|\cdot3M_{2}<3|T|\cdot(M_{1}+M_{2})+|T|\cdot3M_{2}$,
using that $M_{1}>\sum_{\ell=0}^{3n-1}2^{\ell}=2^{3n}-1$. Hence,
solution~$S'$ (as defined above) is better than the considered optimal
solution, which yields a contradiction. Hence, any optimal solution
satisfies the first property.

Suppose now that we are given an optimal solution that violates the
second property. Then there is a job $L(a_{i},b_{j},c_{k})$ that
is not completely scheduled during $[t(a_{i},b_{j},c_{k}),t(a_{i},b_{j},c_{k})+3)$.
Then its contribution for the global profit is $3\cdot M_{2}$ and
this does not depend on where exactly after $t(a_{i},b_{j},c_{k})+3$
it finishes. Thus, we can move $L(a_{i},b_{j},c_{k})$ so that it
is completely scheduled non-preemptively during $[3|T|+3n,\infty)$,
without changing its contribution to the profit.\qed \end{proof}
\begin{lemma}
\label{lem:high-profit=>yes}If there is an optimal solution
of the GPSP instance yielding a profit of at least $|T|\cdot(3\cdot M_{1}+3\cdot M_{2})+|T|\cdot3\cdot M_{2}+2^{3n}-1$,
then the instance of {\sc 3-Dimensional Matching} is a ``yes''-instance.\end{lemma}
\begin{proof}
Given a solution to our defined GPSP instance, by Lemma~\ref{lem:optimal-solution-standard-form}
we can assume it is in standard form. Let $T'\subseteq T$ be the
set of triples $(a_{i},b_{j},c_{k})\in T$ corresponding to the jobs
$L(a_{i},b_{j},c_{k})$ that are scheduled during their respective
intervals $[t(a_{i},b_{j},c_{k}),t(a_{i},b_{j},c_{k})+3)$. We claim
that they form a feasible solution and that $|T'|=n$. The profit
obtained by the jobs $L(a_{i},b_{j},c_{k})$ corresponding to the
triples in~$T'$ equals $|T'|\cdot(3\cdot M_{1}+3\cdot M_{2})$. The
total profit given by all other jobs $L(a_{i},b_{j},c_{k})$ equals
$(|T|-|T'|)\cdot3\cdot M_{2}$. For jobs $A(a_{i},b_{j},c_{k})$,
 $B(a_{i},b_{j},c_{k})$, and $C(a_{i},b_{j},c_{k})$ such that $L(a_{i},b_{j},c_{k})$
is \emph{not }scheduled during $[t(a_{i},b_{j},c_{k}),t(a_{i},b_{j},c_{k})+3)$
we can assume that the former are all scheduled during $[t(a_{i},b_{j},c_{k}),t(a_{i},b_{j},c_{k})+3)$
(since there they yield the maximum profit and no other job is scheduled
there). The profit of those jobs equals then $(|T|-|T'|)\cdot3(M_{1}+M_{2})$.
Denote by $J'$ the set of all other jobs $A(a_{i},b_{j},c_{k})$,
$B(a_{i},b_{j},c_{k})$, and $C(a_{i},b_{j},c_{k})$ (i.e., such that
$L(a_{i},b_{j},c_{k})$ is scheduled during $[t(a_{i},b_{j},c_{k}),t(a_{i},b_{j},c_{k})+3)$).
Their total profit equals $3|T'|\cdot M_{2}+\sum_{A(a_{i},b_{j},c_{k})\in J'}2^{3n-i}+\sum_{B(a_{i},b_{j},c_{k})\in J'}2^{2n-j}+\sum_{C(a_{i},b_{j},c_{k})\in J'}2^{n-k}$.
Thus, the total profit of the solution equals
\begin{eqnarray*}
 &  & |T'|\cdot(3\cdot M_{1}+3\cdot M_{2})+(|T|-|T'|)\cdot3\cdot M_{2}+(|T|-|T'|)\cdot3(M_{1}+M_{2})+3|T'|\cdot M_{2}\\
 &  & +\sum_{A(a_{i},b_{j},c_{k})\in J'}2^{3n-i}+\sum_{B(a_{i},b_{j},c_{k})\in J'}2^{2n-j}+\sum_{C(a_{i},b_{j},c_{k})\in J'}2^{n-k}\\
 & = & |T|\cdot(3\cdot M_{1}+3\cdot M_{2})+|T|\cdot3\cdot M_{2}+\sum_{A(a_{i},b_{j},c_{k})\in J'}2^{3n-i}+\sum_{B(a_{i},b_{j},c_{k})\in J'}2^{2n-j}+\sum_{C(a_{i},b_{j},c_{k})\in J'}2^{n-k}
\end{eqnarray*}
Thus, if the total profit is at least $|T|\cdot(3\cdot M_{1}+3\cdot M_{2})+|T|\cdot3\cdot M_{2}+2^{3n}-1$
then it must be that $\sum_{A(a_{i},b_{j},c_{k})\in J'}2^{3n-i}+\sum_{B(a_{i},b_{j},c_{k})\in J'}2^{2n-j}+\sum_{C(a_{i},b_{j},c_{k})\in J'}2^{n-k}\ge2^{3n}-1$.
Therefore, there must be at least one (and by construction also at
most one) one job $A(a_{i},b_{j},c_{k})\in J'$ with $i=1$, and similarly
for each $\ell\in\{1,\hdots,n\}$, there are at least $\ell$ jobs $A(a_{i},b_{j},c_{k})\in J'$
with $i\le\ell$. This implies that $\sum_{A(a_{i},b_{j},c_{k})\in J'}2^{3n-i}=\sum_{\ell=1}^{n}2^{3n-\ell}=2^{3n}-1-\sum_{\ell=1}^{2n-1}2^{\ell}=2^{3n}-1-2^{2n}+1=2^{3n}-2^{2n}$
and thus $\sum_{B(a_{i},b_{j},c_{k})\in J'}2^{2n-j}+\sum_{C(a_{i},b_{j},c_{k})\in J'}2^{n-k}\ge2^{2n}-1$.

With the same argument we can show that for each $\ell\in\{1,\hdots,n\}$ there are
at least $\ell$ jobs~$B(a_{i},b_{j},c_{k})\in J'$ with $j\le\ell$,
and then that for each $\ell\in\{1,\hdots,n\}$ there are at least $\ell$ jobs
$C(a_{i},b_{j},c_{k})\in J'$ with $k\le\ell$. Altogether, this implies
that for each $i\in\{1,\hdots,n\}$ there is exactly one job $A(a_{i},b_{j},c_{k})\in J'$,
and similarly for each $j\in\{1,\hdots,n\}$ and each $k\in\{1,\hdots,n\}$
there is exactly one job $B(a_{i},b_{j},c_{k})\in J'$ and one job
$C(a_{i},b_{j},c_{k})\in J'$. Thus, the set $T'$ is a feasible solution,
and in particular $|T'|=n$.
\qed
\end{proof}

\noindent
Lemmas~\ref{lem:yes=>high-profit} and \ref{lem:high-profit=>yes}
together show the correctness of the above reduction, and hence complete the proof of Theorem~\ref{thm:GPSP-hardness}.  Since by Lemma~\ref{lem:optimal-solution-standard-form} 
we can restrict ourselves to solutions in standard form and those are non-preemptive, 
our reduction works for the preemptive as well as for the non-preemptive setting.

Using the fact that solutions in standard form schedule all jobs,
we can modify the above transformation to show that also the General
Scheduling Problem (GSP) is $\mathsf{NP}$-hard if the processing
time of each job is either exactly 1 or 3. All we need to do is to
define a cost function $w_{j}$ for each job by setting $w_{j}(t):=3\cdot M_{1}+3\cdot M_{2}-f_{j}(t)$,
which in particular ensures that for no completion time any job has
negative costs. Then the goal is to find a feasible schedule on one machine
to minimize $\sum_j w_j(C_j)$.
\begin{corollary}
The General Scheduling Problem (GSP) is $\mathsf{NP}$-hard, even
if the processing time of each job is either exactly 1 or 3. This
holds in the preemptive setting as well as in the non-preemptive setting.\end{corollary}

\section{Discussion and Future Directions}

\label{sec:discussion} In this paper, we obtained the first fixed-parameter
algorithms for several fundamental problems in scheduling.
We identified key problem parameters that determine the complexity
of an instance, so that when these parameters are fixed the
problem becomes polynomial time solvable.

There are several interesting open questions in the connection of
scheduling and $\mathsf{FPT}$. One important scheduling objective is weighted
flow time (which in the approximative sense is not completely understood
yet, even on one machine). It would be interesting to investigate
whether this problem is $\mathsf{FPT}$ when parametrized by, for
instance, an upper bound on the maximum processing time and the maximum
weight of a job (assuming integral input data). For scheduling jobs non-preemptively on identical machines to minimize the
makespan, in Section~\ref{sec:makespanminimization} we proved the
first $\mathsf{FPT}$ result for this problem. It would be interesting
to extend this to the setting with release dates or to consider as
parameter the number of distinct processing times (rather
than an upper bound on those). When the jobs
have precedence constraints, a canonical parameter would be the width of the partial order given by these constraints,
and it would be interesting to study whether $P|prec|C_{\max}$ admits
$\mathsf{FPT}$ algorithms when parameterized by the partial order width and additionally
by the maximum processing time.
For the basic setting of makespan minimization $P||C_{\max}$, only recently Goemans and Rothvo{\ss}~\cite{GoemansRothvoss2014} managed to give a an algorithm of running time $O(n^m)$ for all numbers $m$ of machines---thus, it seems rather challenging to obtain a fixed-parameter algorithm parameterized by $m$.
For scheduling with rejection, we further suggest the number of distinct rejection costs as a parameter to investigate.
Finally, it would be desirable to understand the kernelizability of
the scheduling problems we considered here.

\medskip{}
\noindent
\textbf{Acknowledgment.} We thank an anonymous reviewer of an earlier
version for suggestions how to remove one parameter in the algorithms in Sect.~\ref{sec:schedulekjobswithelllengths}
and to prove Theorem~\ref{thm:gsponlynumscheduledjobshardness}.

\noindent \newpage
 \bibliographystyle{splncs}
\bibliography{short,scheduling-references}

\appendix
\newpage

\end{document}